\documentclass{article}
\pdfoutput=1

\usepackage{setspace}
\usepackage{lipsum}
\usepackage{algorithm}
\usepackage{times}

\usepackage{comment}
\usepackage{amsmath}
\usepackage{amssymb}
\usepackage{amsthm}
\usepackage{graphicx}
\usepackage{picins}
\usepackage{subfigure}
\usepackage{ulem}
\usepackage{capt-of, blindtext}
\usepackage{titlesec}
\usepackage{geometry}
\titleformat*{\section}{\Large\bfseries}
\titleformat*{\subsection}{\large\bfseries\itshape}

\newtheorem{theorem}{Theorem}

\newtheorem{remark}{Remark}
\newtheorem{assumption}{Assumption}

\newcommand{\G}{\mathcal{G}} 		
\newcommand{\E}{\mathcal{E}} 		
\newcommand{\V}{\mathcal{V}} 		
\newcommand{\N}{\mathcal{N}} 		
\newcommand{\R}{\mathbb{R}}

\newcommand{\D}{\mathbb{D}} 		
\newcommand{\Lbig}{\mathbb{L}}		

\providecommand{\keywords}[1]{\textbf{\textit{Keywords}} #1}

\renewenvironment{abstract}
  {\small\quotation
  {\bfseries\noindent{\abstractname}\par\nobreak\smallskip}}
  {\endquotation}

\begin{document}

{\begingroup
\begin{center}
\textbf{\Large{Continuous-time Proportional-Integral Distributed Optimization for Networked Systems}}\\
\vspace*{1\baselineskip} 
\large{Greg Droge$^a$\footnote{To whom correspondence should be addressed.  E-mail: gregdroge@gatech.edu} and Hiroaki Kawashima$^b$ and Magnus Egerstedt$^a$}\par
\vspace*{1\baselineskip} 
\small{$^a$School of Electrical and Computer Engineering, Georgia Institute of Technology, Atlanta, GA 30332, USA.} \\
\small{$^b$Graduate School of Informatics, Kyoto University,  Kyoto 6068501,  Japan.}\\
\end{center}
\endgroup}




\setstretch{1.0}
\begin{abstract}
\noindent In this paper we explore the relationship between dual decomposition and the consensus-based method for distributed optimization.  The relationship is developed by examining the similarities between the two approaches and their relationship to gradient-based constrained optimization.  By formulating each algorithm in continuous-time, it is seen that both approaches use a gradient method for optimization with one using a proportional control term and the other using an integral control term to drive the system to the constraint set. Therefore, a significant contribution of this paper is to combine these methods to develop a continuous-time proportional-integral distributed optimization method.  Furthermore, we establish convergence using Lyapunov stability techniques and utilizing properties from the network structure of the multi-agent system. \par\vspace*{1\baselineskip} 
\noindent\keywords{\\Distributed Optimization, Multi-agent control, Networked-systems}
\end{abstract}

\setstretch{1.5}

\section{Introduction}
\label{sec:Intro}
In recent years, much attention has been given to the control, optimization, and coordination of multi-agent systems.  Multi-agent systems present many challenging problems, while having applications as diverse as data fusion in sensor networks to multiple robots acting in collaboration, e.g. \cite{Olfati-Saber2007, Mesbahi2010, Shamma2007, Olfati-Saber2004, Jadbabaie2003}.  A significant amount of research has been produced on the design of algorithms to allow agents to collaborate and achieve desirable global properties in a distributed manner, e.g. \cite{Mesbahi2010, Shamma2007, Palomar2010}.  In this paper, we will address the problem of distributed optimization, where agents in the network collaborate to minimize a global cost and each agent uses only locally available information, defined by the network structure.  

While many classifications of distributed optimization algorithms may exist, a distinction has recently been made in both \cite{Wei2012} and \cite{Kvaternik2012} that distributed optimization techniques can be divided between two categories: consensus-based gradient methods and decomposition or primal-dual methods.  The consensus-based approach is characterized by algorithms where, at each time step, every agent takes a gradient step along with an averaging or consensus step to reach agreement on the variables being optimized, e.g. \cite{Kvaternik2012, Wang2010, Nedic2009, Wang2011a}.  In contrast, decomposition methods distributedly reach agreement by exploiting the dual of the problem, e.g. \cite{Terelius2011, Rantzer2007, Rantzer2009}, which requires the added collaborative update of pricing or dual variables, e.g. \cite{Nedic2009a, Luenberger2008, Boyd2004}.  Of particular interest is the decomposition method for multi-agent systems presented in \cite{Terelius2011} along with the gradient-based solution for dual problems first presented in \cite{Arrow1958}.  When combined, these methods allow for a gradient-based multi-agent distributed optimization technique that, for simplicity, we refer to as dual-decomposition.

In this paper, we show that the consensus-based and dual-decomposition gradient algorithms are actually very closely related when examined in context of the underlying constrained optimization problem that is solved by these methods.  Specifically, we formulate both the mentioned dual-decomposition method and the consensus-based method in \cite{Kvaternik2012} in control theoretic terms to draw parallels and gain intuition behind why they can naturally be joined together.  In fact, it will become apparent that dual-decomposition is very closely related to integral (I) control and the consensus method is closely related to proportional (P) control.  Therefore, a significant contribution of this paper is to combine these two methods to form a new, proportional-integral (PI) distributed optimization method.  This formulation will be similar to the PI distributed optimization method introduced in \cite{Wang2010} and extended in \cite{Wang2011a, Gharesifard2012}.  However, due to the fact that we create the PI optimization method from the perspective of the dual-decomposition method, which involves a set of constraints associated with the interconnections of agents, it enables us to form a different type of integral control term and reduce the required communication.

While much of the work on distributed optimization has been developed in discrete-time formulations, which are amenable for implementation, e.g. \cite{Wei2012, Nedic2009, Terelius2011, Lobel2008}, a great deal of work recently has been made in continuous-time \cite{Kvaternik2012, Wang2010, Wang2011a, Rantzer2007, Rantzer2009, Kvaternik2011, Gharesifard2012}.  Continuous-time analysis has proven useful as it allows Lyapunov stability conditions to be directly applied to the update-equations for convergence analysis.  It also allows for an intuitive connection between the optimization algorithm proposed in this paper and proportional-integral control.  Moreover, a discretization of the framework proposed in this paper does not pose a significant contribution.  The proportional element has been evaluated in discrete time in \cite{Nedic2009} and the integral element has been evaluated in discrete time in \cite{Terelius2011, Rantzer2007, Rantzer2009}.  Furthermore, a closely related PI distributed optimization algorithm developed in \cite{Wang2010, Wang2011a, Gharesifard2012} (discussed further in Section \ref{ssec:pi_optimization}) was discretized in \cite{Wang2010}.


The remainder of this paper will proceed as follows.  We first introduce the necessary background material for the analysis of the distributed optimization algorithms, including the problem formulation, the graph-based multi-agent model of the network, and a parallel to gradient-based constrained optimization.  This background will allow us to present both dual-decomposition and the consensus based method in Sections \ref{sec:dualDecomposition} and \ref{sec:consensus}.  These two methods are then combined in Section \ref{sec:PIOptimization} to create a PI distributed optimization method.  We further develop this method by presenting a formulation which is scalable to a larger class of multi-agent systems in Section \ref{sec:scalability}.  The paper is then concluded with some final remarks in Section \ref{sec:conclusion}.

\section{Preliminaries}  
\label{sec:Background}
This section introduces the background information necessary to characterize PI distributed optimization.  It begins with the formulation of the distributed optimization problem that is addressed in this paper.  Following this, the graph-based model of the multi-agent network will be introduced.  Gradient-based constrained optimization is then discussed from a high level viewpoint to develop intuition about the underlying relationship between dual-decomposition and the consensus based method.  As similarities to PI control will become readily apparent, this section ends with a brief introduction of the PI control metrics that are used to compare these methods.

\subsection{Problem formulation}
We address the distributed optimization methods in terms of the problem formulation presented in \cite{Nedic2009} and continued in \cite{Kvaternik2012, Wang2010, Wang2011a, SundharRam2010, Gharesifard2012, Kvaternik2011}.  Specifically, assume that the function to be optimized is a summation of strictly-convex costs, i.e. 
\begin{equation} \label{eq-primalproblem_orig}
   \min_{x} \sum_{i=1}^N f_i(x) ,
\end{equation}
where $x \in \mathbb{R}^n$ is the parameter vector being optimized, $N$ is the number of agents,  and agent $i$ only knows its individual cost, $f_i(x)$.  The individual costs can be derived naturally from a distributed problem as done in \cite{Palomar2010} for resource allocation, or can be ``designed'' as done, for example, in \cite{Keviczky2006, Droge2013}, where a central cost is split into separable components which are then assigned to the individual agents.  

To be able to establish convergence to the global minimum certain convexity assumptions are made on the cost.  Note that a convex function is defined as a function that satisfies:
\begin{equation}
f(\theta x + (1-\theta) y) \leq \theta f(x) + (1- \theta) f(y)
\label{eq:convexDefinition}
\end{equation}
where $0 < \theta < 1$.  A function is strictly-convex if strict inequality holds in (\ref{eq:convexDefinition}) (see, for example, \cite{Boyd2004} for a thorough overview of convex functions and their properties).  The following assumptions about the costs are used throughout the paper:
\begin{assumption}
$f_i(x) : \mathbb{R}^n \rightarrow \mathbb{R}$, $i \in \{1,...,N\}$, are convex, twice continuously differentiable functions and the summation $\sum_{i=1}^N f_i(x)$ is strictly-convex, twice continuously differntiable function. 
\label{as:strictlyconvex}
\end{assumption}
\begin{assumption}
The solution $f^* = \min_{x} \sum_{i=1}^N f_i(x)$ and respective optimal parameter vector, $x^*$, exist and are finite.
\label{as:solutionexistence}
\end{assumption}

\begin{remark}
We note that the differentiability assumption has been relaxed in many of the references to address subgradient optimization.  However, we do not concern ourselves with relaxing this assumption as it does not add to the development of the paper.
\end{remark}

For sake of clarifying the notation, one key point must be stressed.  To perform distributed optimization, each agent will maintain its ``own version'' of the variables, denoted as $x_i \in \R^n$, with the constraint that $x_i = x_j$ $\forall$ $i,j \in \{1, ..., N\}$.  This will allow (\ref{eq-primalproblem_orig}) to be expressed as
\begin{equation}
\label{eq:primalproblem_constrained}
	\min_{x_i, i = 1, ..., N} \sum_{i=1}^N f_i(x_i) .
\end{equation}
$$
\mbox{s.t. } x_i = x_j \mbox{ } \forall \mbox{ } i,j \in \{1, ..., N\}
$$
To perform the optimization in a distributed manner, the equality constraints are relaxed.  Algorithms differ in the manner that they force agents to return to the constraint set.


\subsection{Networked multi-agent systems}
\label{ssec:networkedSystems}
We now introduce the terminology and properties of multi-agent systems that will be used to formulate the distributed optimization algorithms and discuss their convergence.  The term ``agent'' is used to refer to a computing component and it is assumed that agents only communicate with each other through a defined, static network topology.  This is representative of a great number of different multi-agent systems, from communication networks to teams of robots, e.g. \cite{Mesbahi2010, Shamma2007, Rantzer2009}.  

The interconnections of the network are represented through an undirected graph $\G (\V, \E)$.  The set of nodes, $\V$, is defined such that $v_i \in \V$ corresponds to agent $i$.  Communication constraints are represented through the set of edges in the graph, $\E \subseteq \V \times \V$, where $(v_i, v_j) \in \E$ iff agents $i$ and $j$ can directly communicate. The number of agents is then given by $|\V| = N$ and the number of communication links is given by $|\E| = M$.  To prove convergence of the distributed optimization methods, the following assumption on the graph topology is made
\begin{assumption}
The graph $\G(\V, \E)$ is connected.
\label{as:ConnectedGraph}
\end{assumption}

Associated with this graph are two important, and related matrices.  The first is the incidence matrix, $D \in \R^{N \times M}$ which is formed by arbitrarily assigning an orientation to each edge and can be defined as 
\begin{equation}
	D = [d_{ik}] = 
	\begin{cases}
		1 & \mbox{edge } k \mbox{ points to node } i \\
		-1 & \mbox{edge } k \mbox{ originates at node } i \\
		0 & \mbox{otherwise}
	\end{cases}.
\label{eq:indicinceMatrix}
\end{equation}
The second matrix, the graph Laplacian, is closely related to $D$ and can be defined as $L = L^T = DD^T \in \mathbb{R}^{N \times N}$.  Note that the resulting values for the elements of $L$ are independent of the orientation assigned to each edge, \cite{Mesbahi2010}.

We utilize both the incidence matrix and the graph Laplacian to form larger, aggregate matrices to incorporate the fact that each agent will be maintaining an entire vector of values.  First, let $x_{ij}$ denote the $j^{th}$ element of $x_i$, $z_j \triangleq [x_{1j}, x_{2j}, ... , x_{Nj}]^T$ is the combination of all the $j^{th}$ elements, and $z \triangleq [z_1^T,...,z_n^T]^T \in \mathbb{R}^{Nn}$ is the aggregate state vector.  The aggregate matrices can then be written as $\D \triangleq I_n \otimes D$  and $\Lbig \triangleq  I_n \otimes L$.  This notation expresses the concept that an aggregate graph is formed where there are $n$ replicas of $\G$, each corresponding to one of the elements of the vector being optimized.  The aggregate graph will not be connected, but have $n$ connected components, given Assumption \ref{as:ConnectedGraph}.  Therefore, the aggregate Laplacian will have the following properties (see, for example, \cite{Mesbahi2010}):
\begin{enumerate}
\item [(1)] $\Lbig = \Lbig^T = \D \D^T$
\item [(2)] $\Lbig \succeq 0$
\item [(3)] The eigenvectors associated with the zero eigenvalues of the aggregate Laplacian are $\alpha \otimes \mathbf{1}$, where $\alpha \in \mathbb{R}^n$
\item [(4)] If $\dot{z} = -\Lbig z$, the solution, $\bar{z} = z(t)$ as $t \longrightarrow \infty$, will be the projection of $z(0)$ onto the set $\alpha \otimes \mathbf{1}$ for $\alpha \in \mathbb{R}^n$.  Moreover, the vector $-\Lbig z$ will point along a line orthogonal to the set $\{ \alpha \otimes \mathbf{1} | \alpha \in \mathbb{R}^n \}$.
\end{enumerate}

One further property that will be exploited throughout the paper comes from the incidence matrix.  The constraint in (\ref{eq:primalproblem_constrained}) that $x_i = x_j$ $\forall$ $i,j \in \{1, ..., N\}$ can be written as $\D^T z = 0$.  This can be verified by first considering the scalar case where $n=1$ and $\D = D$.  $D^T z = 0$ will enforce that $x_{k_1} - x_{k_2} = 0$, where $k_1$ and $k_2$ correspond to the verticies associated with edge $k$.  Then through Assumption \ref{as:ConnectedGraph}, $x_i = x_j$ $\forall$ $i,j \in \{1, ..., N\}$.  The same argument can be extended to $n>1$ case by noting that: 
$$\D^T z = \begin{bmatrix} D^T z_1 \\ \vdots \\ D^T z_n \end{bmatrix}.$$

Finally, this notation allows the distributed optimization problem to be presented in a compact form:
\begin{equation}
	\begin{split}
	\min_z f(z) \\
	\mbox{s.t. } h(z) = 0.
	\end{split}
\label{eq:constrained_opt}
\end{equation}
where $f(z) = \sum_{i=1}^N f_i(x_i)$ and $h(z) = \D^T z$.

\subsection{PI control as gradient method for constrained optimization}
\label{ssec:constrained_opt}
We now take note of the structure of (\ref{eq:constrained_opt}) to give intuition to the relationship between the gradient methods presented in Sections \ref{sec:dualDecomposition} and \ref{sec:consensus}.  The development in this section will not dwell on the details of constrained optimization, as it has been a widely studied area, e.g. \cite{Luenberger2008, Boyd2004}.  Rather, it will be focused on forming a control law to return the state to the constraint set when the constraints are relaxed. 

Without constraints, a gradient method for optimization of the problem would simply take the form $\dot{z} = -k_G \frac{\partial f}{\partial z}^T$, where $k_G \in \mathbb{R}^+$ is some gain.  However, when the optimization includes constraints, the update to the variables being optimized cannot be in any arbitrary direction.  The update can only occur in a direction that will allow the state to continue to satisfy the constraint.  As the constraints in (\ref{eq:constrained_opt}) are linear, this involves taking the gradient and projecting it onto the constraint space, as shown in Figure \ref{fig:Constrained_Opt}.  


\begin{figure*}[!t]
\begin{minipage}[b]{0.5\linewidth}
\begin{center}
\includegraphics[width=\linewidth]{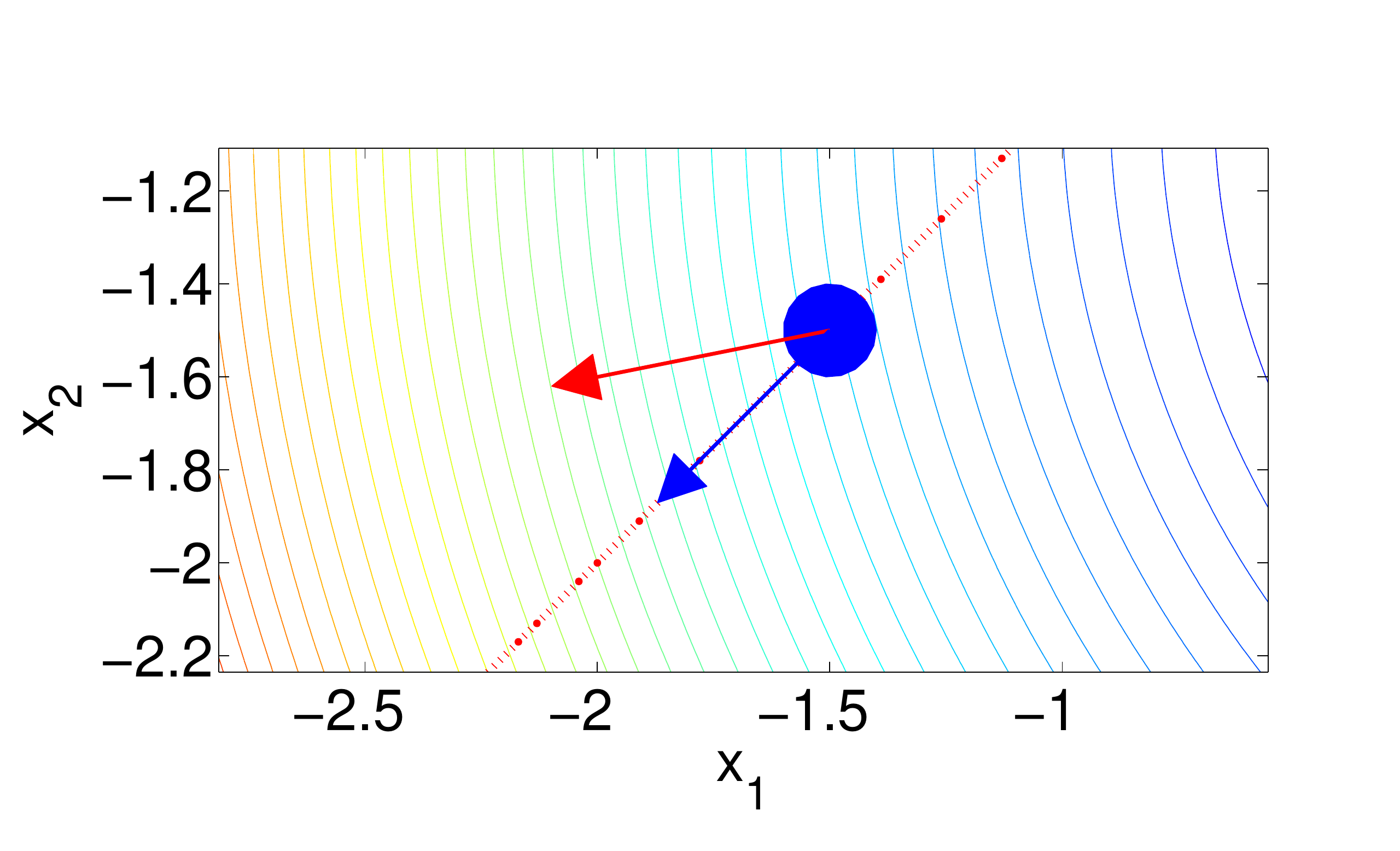} 
\end{center}
\end{minipage}
\begin{minipage}[b]{0.5\linewidth}
\begin{center}
\includegraphics[width=\linewidth]{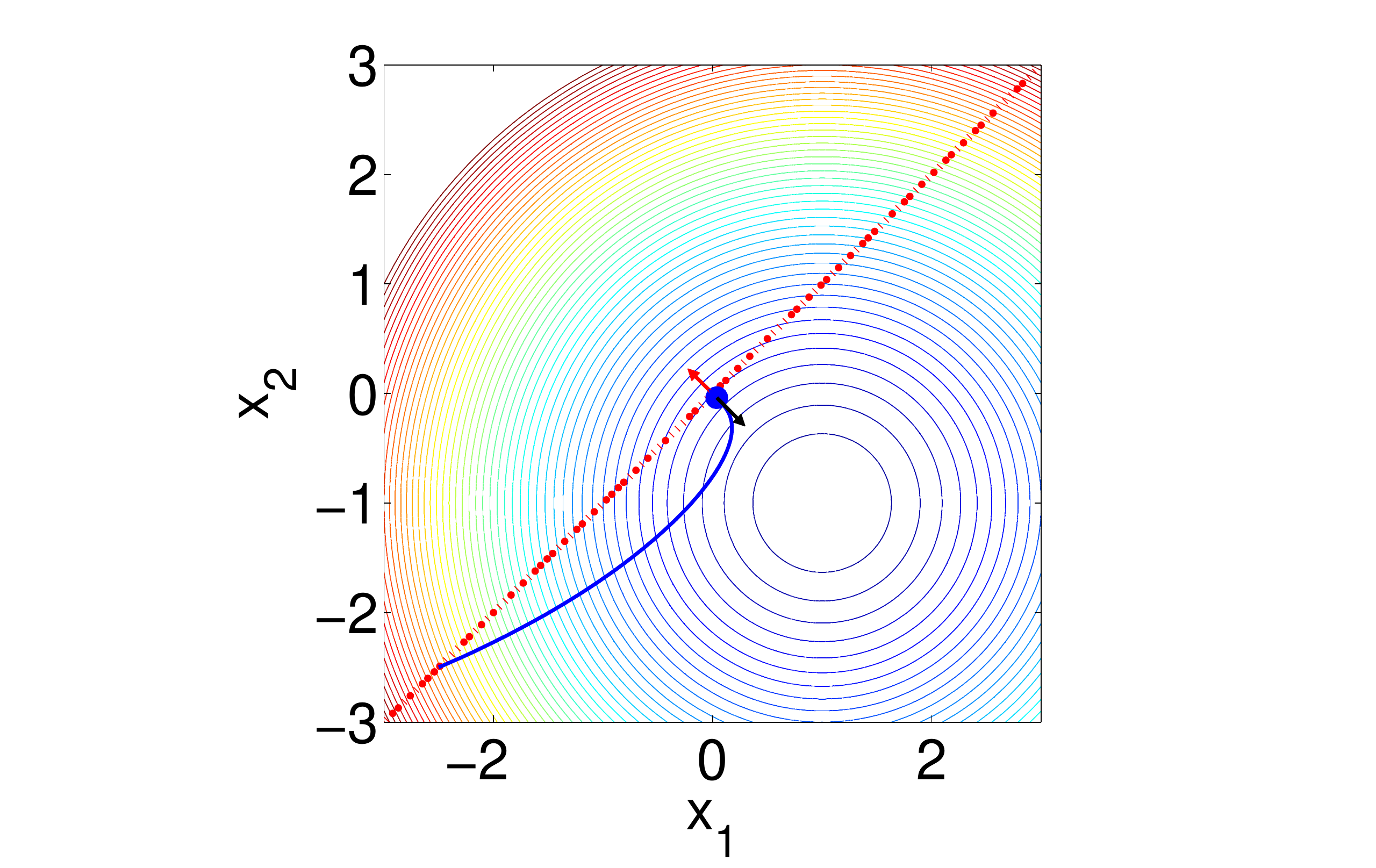} 
\end{center}
\end{minipage}
\caption{This figure shows the results of using the cost $f(z) = (x_1 - 1)^2 + (x_2 + 1)^2$.  Left: Dotted line shows the equality constraint and the arrows show the gradient and projected gradient.  Right: Result of performing the PI gradient method for optimization given in (\ref{eq:constrain_PI}).  The trajectory of the two states is shown ending in the final condition denoted by the solid circle and the constraint is shown as a dotted line.  The arrows show the final gradient and Lagrange multiplier multiplied by the constraint.  As expected, these are equal in magnitude, but opposite in direction. }%
\label{fig:Constrained_Opt}

\end{figure*}

It should be noted that the difference between an unconstrained gradient and a constrained gradient could be written in terms of the addition of a term perpendicular to the constraint set.  This could be expressed as $\frac{\partial f}{\partial z} + \lambda^T \frac{\partial h}{\partial z} = \frac{\partial f}{\partial z} + \lambda^T \D^T$.  The dynamics of the resulting optimization would then be
\begin{equation}
	\dot{z} = -k_G (\frac{\partial f}{\partial z}^T + \D \lambda(t)).
\end{equation} 
However, computing $\lambda(t)$ in a distributed fashion could be difficult as it may require knowledge from the entire network.  

Alternatively, if the gradient method is permitted to violate the constraint, control terms can be added to guide the state back to the constraint at the optimal point.  The first term we consider is a term proportional to the error from the constraint.  Allow $\lambda(t) = \frac{k_P}{k_G} e(t)$ where $e(t) = \D^T z(t)$ is the error at each edge of the graph.  This can be seen to be a logical choice because, as mentioned in Section \ref{ssec:networkedSystems}, $-\D e(t) = -\Lbig z(t)$ will point along a line orthogonal to the constraint set.  In other words, it points in the right direction, but with possibly the wrong magnitude.  This gives the dynamics
\begin{equation}
	\dot{z} = -k_G \frac{\partial f}{\partial z}(z) - k_P \D e(t) = -k_G \frac{\partial f}{\partial z}(z) - k_P \Lbig z(t).
	\label{eq:proportional_control}
\end{equation}
As will be discussed in Section \ref{sec:consensus}, the similarity of (\ref{eq:proportional_control}) to proportional control is perpetuated in that the steady-state solution will have a constant error from the desired optimal point.  Basically, the effort produced by introducing an error term proportional to the deviation from the constraint will fall short of the needed effort to drive the state all of the way to the constraint set.

To compensate for the steady state error, it is common to add an integral term to the control, e.g. \cite{GF2001}.  This would lead to a $\lambda(t)$ of the form $\lambda(t) = \frac{k_P}{k_G}e(t) + \frac{k_I}{k_G} \int_{t_0}^t e(s)ds$.  Over time, the integral term will build up the necessary effort to reach the constraint.  With this additional term, the dynamics of the system can be expressed as 
\begin{equation}
	\dot{z} = -k_G \frac{\partial f}{\partial z}^T - \D \Bigl( k_P e(t) + k_I \int_{t_0}^t e(s)ds \Bigr)
	        = -k_G \frac{\partial f}{\partial z}^T - k_P \Lbig ^T z - k_I \Lbig \int_{t_0}^t z(s)ds .
\label{eq:constrain_PI}	        
\end{equation}
It will be shown in Section \ref{sec:PIOptimization} that under Assumptions \ref{as:strictlyconvex}, \ref{as:solutionexistence}, and \ref{as:ConnectedGraph}, the dynamics in (\ref{eq:constrain_PI}) will indeed converge to a global minimum, as shown in Figure \ref{fig:Constrained_Opt}.  

While this method for obtaining a gradient strategy to arrive at the desired optimal value may seem somewhat trivial or ad-hoc, it will be seen in Section \ref{sec:dualDecomposition} that the dual-decomposition distributed optimization method will exactly correspond to adding an integral term.  Similarly, in Section \ref{sec:consensus}, it is shown that the consensus-based method will be exactly the proportional term.  Therefore, we combine the two methods in Section \ref{sec:PIOptimization} to form a PI distributed optimization method.

\subsection{PI performance metrics}
As the distributed control laws developed throughout the remainder of this paper are closely related to proportional and integral control laws, we give a brief introduction to the performance metrics that will be employed for comparison.  These metrics are important as there really is no single metric which best determines which control law is most suitable.  For example, as discussed in \cite{GF2001}, proportional control can converge quickly, but may result in a steady-state error.  As the proportional gain is increased, the steady-state error will typically decrease up to the point where the system becomes unstable.  On the other hand, integral control can be introduced to eliminate steady-state error, but dampening will be decreased and this will result in greater oscillation, overshoot, and slower convergence.

Therefore, to say that one method is ``better'' would required a reference to a specific application.  To be able to judge which method is more suitable for the given application, the following performance metrics, typical for classic control evaluation (e.g. \cite{GF2001}), are used:  
\begin{itemize}
\item Percent overshoot ($M_p$): The percentage of the distance that the state overshoots the final value, given as $\frac{x_{max} - x_f}{x_f - x_0} \times 100$.
\item Settling time ($t_{10}$ and $t_1$): Time it takes for the state to converge to within 10 percent and 1 percent of the final value.  For example, $t_{10}$ is the smallest $t$ such that $.9 \frac{x_f}{x_f - x_0} < x(t) < 1.1 \frac{x_f}{x_f - x_0}$ $\forall$ $t > t_{10}$.
\item Percent error (\% error): The percentage of error from the optimal value ($\frac{|x^* - x_f|}{x_f - x_0} \times 100$).
\end{itemize}
where $x_0$ is the initial value, $x_f$ is the final value, and $x_{max}$ is the maximum value reached.  For simplicity, we have assumed $x_{max} \geq x_f > x_0$.  As these values are measures of scalar states, the worst case over all agents will be presented in each evaluation.

Also note that numerical results depend upon the value of the gains and initial conditions.  To allow for a fair comparison between examples throughout the paper, all gains ($k_G, k_P,$ and $k_I$) are assigned a value of $1$.  Similarly, all initial conditions are assigned a value of $0$, unless otherwise stated.

\section{Dual decomposition}
\label{sec:dualDecomposition}
This section introduces the concept of gradient-based distributed optimization through the introduction of dual-decomposition, which has been used in a variety of different applications, e.g. \cite{Palomar2010, Wang2010, Terelius2011, Droge2013, Rantzer2009, Giselsson2010}.  Notation, examples, and proofs are given which will allow for a concise development of the distributed optimization methods in Sections \ref{sec:consensus} and \ref{sec:PIOptimization}.  

As already mentioned, dual-decomposition will be akin to integral control for constrained optimization.  However, to provide intuition as to the origins and the theoretical underpinnings of this method, it is presented here in a more typical fashion relying upon the theory of dual-optimization, e.g. \cite{Luenberger2008, Boyd2004}.  The formulation introduced here is closely related to that found in \cite{Terelius2011}, except that we use Uzawa's saddle point method, \cite{Arrow1958}, to update both the parameters and dual variables simultaneously.  This permits a continuous-time formulation where Lyapunov methods can be readily applied to establish convergence.   
After presenting the algorithm, the relation to integral control will be evaluated.  This section will end with a distributed implementation and a numerical example.

\subsection{Dual-decomposition for networked systems}
The basic idea behind dual decomposition is to introduce $n$ copies of the variables, with the constraint that the copies be equal.  The dual is then formed to relax the added constraints and a $\max \min$ optimization technique is used to solve the dual problem.  In this paper, we use a gradient method introduced in \cite{Arrow1958} for saddle point finding.  This will allow for a distributed solution to the problem where each agent uses only local information defined by the network graph, $\G$.

The dual problem to (\ref{eq:primalproblem_constrained}) can be formed by introducing a Lagrange multiplier vector, $\mu_k \in \mathbb{R}^n$, $k = 1, ..., M$, for each edge in $\G$.  It can be written as:
\begin{equation} \label{eq-dual-original}
 \max_{\mu_k, k=1,...,M} \min_{x_i, i = 1, ..., N} \left\{k_G \sum_{i=1}^N f_i(x_i) + k'_I\sum_{k=1}^M \mu_k^T (x_{k_1} - x_{k_2})\right\}
\end{equation}
where, again, the subscripts $k_1$ and $k_2$ correspond to the agents which make up the $k^{th}$ edge and $k_G,k'_I > 0$ are constant gains.  Note that due to the constraint equaling zero, $k'_I$ has no influence and $k_G$ scales the cost, but does not change the location of the optimal point.  Equation (\ref{eq-dual-original}) can be simplified by forming an aggregate Lagrange multiplier vector, $\mu \in \mathbb{R}^{Mn}$, in the same fashion that the aggregate state, $z$, was formed.  This allows us to reintroduce the constraint as $\D^T z = 0$ and rewrite (\ref{eq-dual-original}) as:
\begin{equation} \label{eq-dual-reform}
 \max_{\mu} \min_{z} F(z, \mu) =  k_G f(z) + k'_I z^T \D \mu.
\end{equation}

To solve this max-min problem, we use a technique first developed in \cite{Arrow1958} for saddle point finding and has more recently gained attention for its applicability to distributed optimization, e.g. \cite{Wang2010, Wang2011a, Rantzer2009, Rantzer2007}.  The basic idea behind this approach is that dynamics can be assigned to the variables being optimized and convergence can be established using control methods such as Lyapunov stability.  

For a saddle point finding problem, where $F(z,\mu)$ is strictly-convex in $z$ and strictly concave in $\mu$, \cite{Arrow1958} shows that applying the dynamics
\begin{align}
 \dot{z} = -\frac{\partial F}{\partial z}^T,  \qquad \dot{\mu} = \frac{\partial F}{\partial \mu}^T,
 \label{eq:Uzawa_alg}
\end{align}
the system will converge asymptotically to the saddle point.  Taking the partials of (\ref{eq-dual-reform}), the dynamics can be expressed as:
\begin{equation}
	\dot{z} = -k_G\frac{\partial f}{\partial z}^T - k_I' \D \mu 
\label{eq-dual_decomp_z_dynamics}
\end{equation}
\begin{equation}
	\dot{\mu} =  k_I' \D^T z.
\label{eq-dual_decomp_mu_dynamics}
\end{equation}
However, we note that (\ref{eq-dual-reform}) is not strictly concave in $\mu$, rather, it is linear.  This requires further evaluation, which is done in the proofs of Theorems \ref{th-dual_decomp} and \ref{th-dual_global_min}.  While there exist proofs for dual-decomposition, e.g. \cite{Palomar2010, Feijer2010}, we present an alternative proof here to show the relationship of dual-decomposition to the underlying constrained optimization problem.  This will allow us to easily extend these proofs in Section \ref{sec:PIOptimization} for the PI distributed optimization method that will be developed.  The proofs use the same Lyapunov candidate function as \cite{Rantzer2007, Feijer2010} to prove convergence, but differ in the application of Lasalle's invariance principal and the proof that the equilibrium reached is the global minimum.

\begin{theorem}  \label{th-dual_decomp}
Given Assumptions \ref{as:strictlyconvex}, \ref{as:solutionexistence}, and \ref{as:ConnectedGraph} as well as the dynamics in (\ref{eq-dual_decomp_z_dynamics}) and (\ref{eq-dual_decomp_mu_dynamics}), the saddle point $(\dot{z}, \dot{\mu}) = (0,0)$ is globally asymptotically stable.
\end{theorem}
\begin{proof}
Using the candidate Lyapunov function $V = \frac{1}{2}(\dot{z}^T \dot{z} + \dot{\mu}^T \dot{\mu} )$, $\dot{V}$ can be written as:
\begin{equation}
\begin{split}
\dot{V} &= \dot{z}^T \ddot{z} + \dot{\mu}^T \ddot{\mu} = -\dot{z}^T \bigr(k_G \frac{\partial^2 f}{\partial z^2} \dot{z} + k_I'\D \dot{\mu} \bigl) + k_I'\dot{\mu}^T \D^T \dot{z} \\
&=-k_G\dot{z}^T \frac{\partial^2 f}{\partial z^2} \dot{z} = -\dot{z}^T H(z) \dot{z} \leq 0 \mbox{ } \forall \dot{z}, \dot{\mu}
\end{split}
\end{equation}
where $H(z) = k_G \frac{\partial^2 f}{\partial z^2} \succ 0$ due to strict convexity given by Assumption \ref{as:strictlyconvex}.
As there is no dependence upon $\dot{\mu}$ in $\dot{V}$, LaSalle's invariance principle must be used to show convergence to $(\dot{z}, \dot{\mu}) = (0,0)$.

Let the set where $\dot{V} = 0$ be denoted as
\begin{equation}
S = \{ (\dot{z}, \dot{\mu}) | \dot{V} = 0 \} = \{ (\dot{z} = 0, \dot{\mu} \in \mathbb{R}^{Mn}) \}
\end{equation}
To see that that the only solution in which the complete state $(\dot{z}, \dot{\mu})$ can remain in $S$ is the equilibrium $(0,0)$, use the fact that to stay in $S$ $\Rightarrow \dot{z} = 0$ $\forall t$ $\Rightarrow \ddot{z} = 0$.  From this we see that
$$
\ddot{z} = -H(z) \dot{z} - k_I'\D \dot{\mu} = -k_I \D \D^T z = -k_I\Lbig z = 0,
$$
where $k_I'^2 = k_I$.  For the connected graph, the only $z$ such that $-\Lbig z = 0$ is $z =  \alpha \otimes \mathbf{1}$, $\alpha \in \mathbb{R}^n$.  This shows two things:
\begin{enumerate}
\item [(1)] $x_{i} = x_{j}$ $\forall i, j$ which means that the agents reach consensus.
\item [(2)] $\dot{\mu} = k_I' \D^T(\alpha \otimes \mathbf{1}) = 0$ which shows that the only possible value for $\dot{\mu}$ which stays in $S$ is $\dot{\mu} = 0$.
\end{enumerate} 

Since $V$ is radially unbounded, this completes the proof.
\end{proof}

\begin{theorem} \label{th-dual_global_min}
Given Assumptions \ref{as:strictlyconvex}, \ref{as:solutionexistence}, and \ref{as:ConnectedGraph} as well as the dynamics in (\ref{eq-dual_decomp_z_dynamics}) and (\ref{eq-dual_decomp_mu_dynamics}), the saddle point $(\dot{z}, \dot{\mu}) = (0,0)$ corresponds to the global minimum. 
\end{theorem}

\begin{proof}
To validate that a feasible solution is a local extremum, $z^*$, of a constrained optimization problem {\it it is sufficient to show that $z^*$ corresponds to a regular point (i.e. rows of $\frac{\partial h}{\partial z}(z^*)$ are linearly independent) and there exists $\lambda^*$ such that  
\begin{equation}
0 = \frac{\partial f}{\partial z}(z^*) + {\lambda^*}^T \frac{\partial h}{\partial z}
\label{eq:equality_condition}
\end{equation}
where $h(z) = 0$ is the constraint and $f(z)$ is the cost}, (see \cite{Luenberger2008} for a discussion on local extremum and regular points).  Due to Assumption \ref{as:strictlyconvex}, the only extremum is the global minimum.  Therefore, this proof is performed in two steps.  First, we show that the saddle point corresponds to a feasible point satisfying (\ref{eq:equality_condition}), then we show that the saddle point is indeed a regular point.  

The proof of Theorem \ref{th-dual_decomp} showed that $(\dot{z}, \dot{\mu}) = (0,0)$ implies that consensus is reached.  Thus, the constraints are satisfied and the saddle point is feasible.  Also, by noting that $\frac{\partial h}{\partial z} = \D^T$ for the problem at hand, (\ref{eq-dual_decomp_z_dynamics}) gives us
\begin{equation}
\dot{z} = 0 = -k_G \frac{\partial f}{\partial z}^T - k_I' \frac{\partial h}{\partial z}^T \mu \Rightarrow 0 = k_G \frac{\partial f}{\partial z} + k_I' \mu^T \frac{\partial h}{\partial z} .
\label{eq:minCond_from_zdot}
\end{equation}  
Allowing $\lambda = \frac{k_I'}{k_G} \mu $, (\ref{eq:equality_condition}) is satisfied.  

The saddle point must now be shown to be a regular point.  To do so, we show that the convergent point is a regular point to the problem in which edges are removed from $\G$ to form a minimum spanning tree (for undirected graphs, a minimum spanning tree is a connected graph with $N$ nodes and $N-1$ edges, e.g. \cite{Mesbahi2010}).  Due to Assumption \ref{as:ConnectedGraph}, a minimum spanning tree, $\G_T$, exists such that $\E_T \subset \E$.  The saddle point is shown to be regular by first showing that the representation of the constraints using $\G_T$, i.e. $\D^T_T z = 0$, is linearly independent and then showing that if a $\lambda$ can be found to satisfy (\ref{eq:equality_condition}) for $\G$, a $\lambda_T$ can be found to satisfy (\ref{eq:equality_condition}) for $\G_T$.  

Let $D_T \in \R^{N \times N-1}$ be the incidence matrix associated with $\G_T$.  The graph Laplacian for a connected graph with $N$ nodes always has rank $N-1$, \cite{Mesbahi2010}.  Therefore, $\D_T$ has full rank, which for $n = 1$, gives that $D_T^T z = 0$ is a linearly independent set of constraints.  For $n > 1$, $\D_T = I_n \otimes D_T$, and, as noted in Section \ref{ssec:networkedSystems}, $\D_T^T z = \begin{bmatrix}D_T^T z_1 \\ \vdots \\ D_T^T z_n \end{bmatrix}$ which will also be linearly independent.

Without loss of generality, we can assume that $D = \begin{bmatrix} D_T & D_R \end{bmatrix}$ where $D_R$ containts the ``redundant'' edges not contained in $\G_T$.  Since $D_T$ has the same rank as $D$, the columns in $D_R$ can be expressed as linear combinations of the columns of $D_T$.  In other words, $D_R = D_T \delta$, where $\delta \in \R^{N-1 \times M-N+1}$.

Without loss of generality, assume the elements in $z$ have been rearranged to write $D = \begin{bmatrix} D_T & D_R \end{bmatrix}$, where $\D_R = I_n \otimes D_R$.  Since $D_R = D_T \delta$, $\D_R$ can be expressed as $\D_T \Delta$, where $\Delta = \mathbf{1} \otimes \delta$.  We can separate $\lambda$ as $\lambda = \begin{bmatrix} \lambda' \\ \lambda''\end{bmatrix}$ which allows us to write $\D \lambda = \D_T \lambda' + \D_R \lambda'' = \D_T \lambda' + \D_T \Delta \lambda''$.  Therefore, if a $\lambda$ is found such that (\ref{eq:equality_condition}) is satisfied for $\G$, $\lambda_T$ can be defined as $\lambda_T = \lambda' + \Delta \lambda''$.  Thus, the solution is a regular point for the constraint $\D_T^T z = 0$.
\end{proof}


\subsection{Integral control}
With the optimization framework in hand, the loop can be closed on the discussion begun in Section \ref{ssec:constrained_opt} by relating the dynamics in (\ref{eq-dual_decomp_z_dynamics}) and (\ref{eq-dual_decomp_mu_dynamics}) to integral control.  Note that the Lagrange multiplier, $\mu$, can be expressed as follows (assuming $\mu(t_0) = 0$):
\begin{equation}
\mu(t) = \int_{t_0}^t \dot{\mu}(\tau) d \tau = \int_{t_0}^t k_I' \D^T z(\tau)d\tau =  k_I' \D^T \int_{t_0}^t z(\tau)d\tau.
\label{eq:mu_integral_control}
\end{equation}
This allows $\dot{z}$ to be expressed as:
\begin{equation}
\dot{z}(t) = -k_G \frac{\partial f}{\partial z}^T - k_I \D \D^T \int_{t_0}^t z(\tau)d\tau  = -k_G \frac{\partial f}{\partial z}^T - k_I \Lbig \int_{t_0}^t z(\tau)d\tau,
\end{equation}
which gives the same result obtained in (\ref{eq:constrain_PI}) assuming $k_P = 0$.  After closer inspection of (\ref{eq-dual_decomp_mu_dynamics}), one can see that the Lagrange multiplier, $\mu$ is indeed the integral of the weighted error referred to in Section \ref{ssec:constrained_opt}.

\subsection{Distributed implementation}
While the analysis of this method has been performed from the point of view of the entire system, its utility as a distributed optimization technique would be questionable if it were not possible for the algorithm to be executed by each agent using only local information.  Therefore, we now present the algorithm in terms of implementation of a single agent and discuss the information and communication requirements.

Equations (\ref{eq-dual_decomp_z_dynamics}) and (\ref{eq-dual_decomp_mu_dynamics}) can be written in terms of execution by a single agent, $i$, as follows:
\begin{equation}
\dot{x}_i = -k_G \frac{\partial f_i}{\partial x}^T(x_i) - k_I'\sum_{j \in \N_i} \mu_i^j,
\label{eq:x_singleAgent_I}
\end{equation}
\begin{equation}
\dot{\mu}_i^j = k_I'(x_{i} - x_{j}),
\label{eq:mu_singleAgent_I}
\end{equation}
where for simplification we have introduced the Lagrange multiplier variables $\mu_i^j = -\mu_j^i = d_{i,k_{ij}} \mu_{k_{ij}}$ where $k_{ij}$ is the edge connecting agents $i$ and $j$ and it is assumed that $\mu_k(0) = 0,$ $k = 1, ..., M$.\footnote{By uniqueness of solutions to differential equations, $\dot{\mu}_i^j(t) = -\dot{\mu}_j^i(t)$ $\forall t$}  Note that $\N_i$ denotes agent $i$'s neighborhood set, or agents with which agent $i$ can communicate.  By inspection, agent $i$ can compute $\dot{x}_i$ and $\dot{\mu}_i^j$ $\forall$ $j \in \N_i$ using only its own state and the states of its neighbors.  Therefore, we emphasize that {\it the only piece of information that an agent needs to communicate with its neighbors is its version of the state vector}.

\begin{figure*}[!t]
\begin{center}
\includegraphics[width=0.3\linewidth]{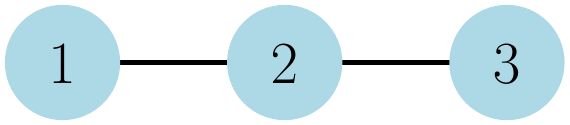} 
\end{center}
\caption{This figure depicts the ``Line'' network structure used for the examples in Sections \ref{sec:dualDecomposition}, \ref{sec:consensus}, and \ref{sec:PIOptimization} }%
\label{fig:Line}
\end{figure*}

\subsection{Example}
To illustrate behaviors typical of dual decomposition, we give a numerical example. Let the individual costs be defined as follows:
\begin{equation}
\begin{split}
f_1(x_1) = (x_{11} - 1)^2 + \frac{1}{3} (x_{11} - x_{12})^2, \\
f_2(x_2) = (x_{22} - 3)^2 + \frac{1}{3} (x_{21} - x_{22})^2, \\
f_3(x_3) = (x_{31} - 6)^2 + \frac{1}{3} (x_{31} - x_{32})^2.
\end{split}
\label{eq:convex_example}
\end{equation}
where $x_i = \begin{bmatrix}x_{i1} & x_{i2} \end{bmatrix}^T$ and the network structure takes the form of the line graph shown in Figure \ref{fig:Line}.  In other words, agents 1 and 2 as well as 2 and 3 can communicate, but agents 1 and 3 cannot.  The global cost is given by $\sum_{i = 1}^3 f_i(x_i)$, where $x_1 = x_2 = x_3$, has the optimal solution of $x^* = \begin{bmatrix}3.4 & 3.2 \end{bmatrix}^T$.  

Figure \ref{fig:convex_I} and Table \ref{tbl:convex} show the results of employing these dynamics.  As seen in Figure \ref{fig:convex_I}, there is oscillation in the solution as the different agents communicate and vary their values.  This oscillation is quite typical of dual-decomposition \cite{Rantzer2007}, and it will be seen that the oscillation increases with an increase in problem complexity and number of agents in Section \ref{sec:scalability}.
Table \ref{tbl:convex} shows that the I control (corresponding to dual-decomposition) has a large overshooot and slower settling times when compared with the P and PI control laws (which are discussed in Sections \ref{sec:consensus} and \ref{sec:PIOptimization}).  This is to be expected as the integral term will decrease the dampening of the system \cite{GF2001}.  Moreover, as expected, Table \ref{tbl:convex} shows that there is zero steady-state error when using dual decomposition.

\begin{table}[t]
\begin{center}
\caption{The results of performing proportional, integral, and PI distributed optimization for the convex optimization problem } 
\label{tbl:convex}
\begin{tabular}{|c|c|c|c|c|}
\hline 
& P: $\gamma = 1$ & P: $\gamma = \frac{1}{1 + .1t}$ & I & PI \\ 
\hline 
$M$ & 0.11\% & 34.66\% & 24.24\% & 14.95\% \\ 
\hline 
$t_{10}$ & 3.54 & 103.73 & 5.61 & 5.14 \\ 
\hline 
$t_1$ & 6.66 & 869.32 & 15.04 & 13.19 \\ 
\hline 
\% error & 43.58\% & 1.97\% & 0\% & 0\% \\ 
\hline 
\end{tabular}
\end{center} 
\end{table}
\begin{figure}%
\centering
\includegraphics[width=.5\linewidth]{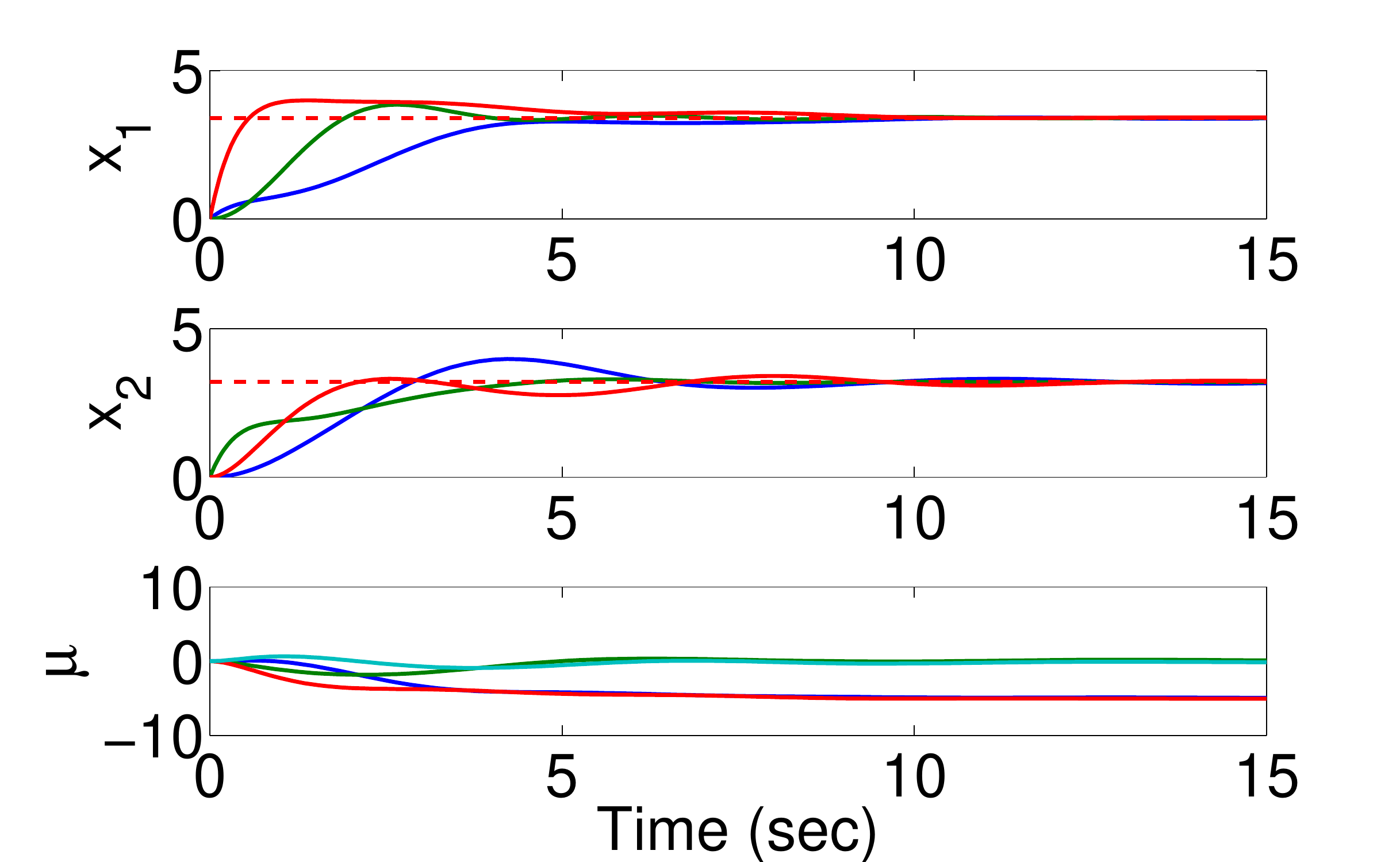} 
\caption{This figure shows the results from the convex optimization example using dual-decomposition}%
\label{fig:convex_I}
\end{figure}

\section{Consensus based distributed optimization}
\label{sec:consensus}
This section introduces the consensus-based distributed optimization technique, first outlined in  \cite{Nedic2009}, which will give the proportional component in the new PI distributed optimization method.  After formulating the algorithm in terms of notation presented in previous sections, characteristics of the convergence are discussed in terms of the constrained optimization problem. This section will end by resuming the example started in the previous section to present a comparison between the distributed optimization methods.

\subsection{Consensus based algorithm}

While originally given in discrete time, we present the consensus based distributed optimization problem in continuous time as done in \cite{Kvaternik2012} to maintain consistent notation.  In stark contrast to the development of dual-decomposition, the consensus-based method was not designed from existing optimization methods.  Rather, it was directly developed for networked, multi-agent systems.  The foundation of this concept is that the consensus equation, a core equation in many multi-agent designs, e.g. \cite{Mesbahi2010, Olfati-Saber2004, Jadbabaie2003}, can be used to force agreement between different agents.  Therefore, the basic idea is for each agent to combine a step in the gradient direction with a step in the direction of consensus.

As the consensus method was developed for the multi-agent scenario, it can immediately be expressed in a distributed fashion as   
\begin{equation}
\dot{x}_i = - k_G \frac{\partial f_i}{\partial x}(x_i) - \sum_{j \in \N_i} \alpha_{ij}(x_i - x_j) ,
\label{eq:consensus}
\end{equation}
where $\alpha_{ij}$ is the weighting that agent $i$ associates with the edge of the graph connecting itself to agent $j$.  Assuming equal weighting on all edges, i.e. $\alpha_{ij} = k_P$ $\forall$ $(v_i, v_j) \in \E$, the consensus based method can be stated for the aggregate state dynamics as:
\begin{equation}
\dot{z} = - k_P \Lbig z - k_G \frac{\partial f}{\partial z}^T.
\label{eq:consensus_based}
\end{equation}
From this expression of the aggregate dynamics, we immediately see that the consensus term is the proportional term given in (\ref{eq:constrain_PI}).  

We do not present a proof of this method as it does not add to the development in this paper.  For the discrete-time analog to (\ref{eq:consensus_based}), using a diminishing or adaptive step-size rule\footnote{Section \ref{sec:consensus} is the only section which consideres the gain $k_G$ to be time-varying.  Throughout the rest of the paper, all gains ($k_G, k_I, $ and $k_P$) are considered constant.} for determining $k_G$ at each iteration of the optimization would cause the agents to converge to the optimal value.  For the continuous case, \cite{Kvaternik2012} proves that agents can come arbitrarily close to the optimum by choosing $\frac{k_G}{k_P}$ to be ``sufficiently small.''

%

The diminishing step-size condition has been observed to be a possible deterrent of quick convergence of the algorithm, e.g. \cite{Wang2010, Nedic2009, Wang2011a}.  To balance a tradeoff between convergence and optimality, \cite{Nedic2009} proposed a scheme of changing $k_G$ during execution to get closer to the optimal point.  The basic idea is that a constant gain often will result in the state approaching a steady-state value in relatively few steps.  Once the state is ``close enough'' to the steady-state then the gain is changed to zero to allow the agents to reach consensus.  They prove that the longer the agents wait to switch to the zero gain, the closer they will come to the optimal value, but will suffer in convergence rate.  

\subsection{Consensus method and constrained optimization}

We now examine this tradeoff further in terms of the underlying constrained optimization problem given in Section \ref{ssec:constrained_opt}.  This will give insight into the effect of the contribution of the proportional term and the benefit of including an integral term, which is done in Section \ref{sec:PIOptimization}.  

To perform this analysis, assume that $\bar{z}$ is the steady-state result of executing (\ref{eq:consensus_based}) as $t \longrightarrow \infty$.  Such a $\bar{z}$ is known to exist due to the analysis in \cite{Kvaternik2012}.  At $\bar{z}$, (\ref{eq:consensus_based}) will give
\begin{equation}
\dot{z} = 0 = - k_P \Lbig \bar{z} - k_G \frac{\partial f}{\partial z}^T(\bar{z}) \Rightarrow 0 = \frac{k_P}{k_G} \Lbig \bar{z} + \frac{\partial f}{\partial z}^T(\bar{z}).
\label{eq:convergent_dynamics}
\end{equation}

Using the fact that $\Lbig = \Lbig^T = \D \D^T$, (\ref{eq:convergent_dynamics}) can be expressed as $0 =  \frac{\partial f}{\partial z} + \frac{k_P}{k_G} \bar{z}^T \D \D^T$.  Now, let $\lambda^T = \frac{k_P}{k_G} \bar{z}^T \D$ and recall that $\frac{\partial h}{\partial z} = \D^T$, where $h(z) = 0$ is the equality constraint.  This gives $0 = \frac{\partial f}{\partial z} + \lambda^T \frac{\partial h}{\partial z}$ as in (\ref{eq:equality_condition}).  While this satisfies part of the condition for determining an extreme point, $\bar{z}$ will not be optimal as consensus will not be reached, resulting in the constraints not being met, \cite{Kvaternik2012}.  

As discussed in Section \ref{ssec:networkedSystems}, $\Lbig z$ will always point along lines perpendicular to the constraint set.  This means that $\bar{z}$ will be a point where $\frac{\partial f}{\partial z}(\bar{z})$ points along a line perpendicular to the constraint set.  Now, let $\bar{z}' = z(t)$ as $t \longrightarrow \infty$ where $\dot{z} = - \Lbig z$ and $z(0) = \bar{z}$.  Since $\Lbig z$ points directly to the constraint set, $\bar{z}'$ will be the point of intersection of the constraint set orthogonal to $\bar{z}$.  Therefore, if $f(z)$ is such that the gradient will always point directly at the unconstrained optimal point, then the result of the optimization strategy proposed in \cite{Nedic2009} can converge arbitrarily close to the optimal value.  An example of such a convex function is shown in Figure \ref{fig:Constrained_Opt}.

More important to our discussion is that a constantly weighted consensus term will not have enough control authority to pull the state of the system all of the way to the optimal point. However, it will help to guide the state to, and maintain it on, a line in which the only additional control effort need be in the direction of consensus.  This further motivates the choice of adding an integral control term.

\subsection{Example}
We continue the example started in Section \ref{sec:dualDecomposition} using the consensus-based distributed optimization.  Two scenarios are shown for the gain: $k_G = 1$ which violates the diminishing or adaptive gain requirement and $k_G = \frac{1}{1 + .1 t}$ which satisfies the requirement.  The results are shown in Figure \ref{fig:convex_P} and Table \ref{tbl:convex}.  The constant gain example exhibits the very desirable attribute of quick convergence, however suffers in performance as the values do not converge and the optimal value is not reached.  On the other hand, the fading gain example shows that the optimal values can be achieved, but convergence suffers as expected.  Both exhibit the desirable attribute of very little oscillation in the solution, however, the fading gain does show a significant increase in overshoot.

\begin{figure*}[!t]
\begin{center}$
\begin{array}{cc}
\includegraphics[width=.45\linewidth]{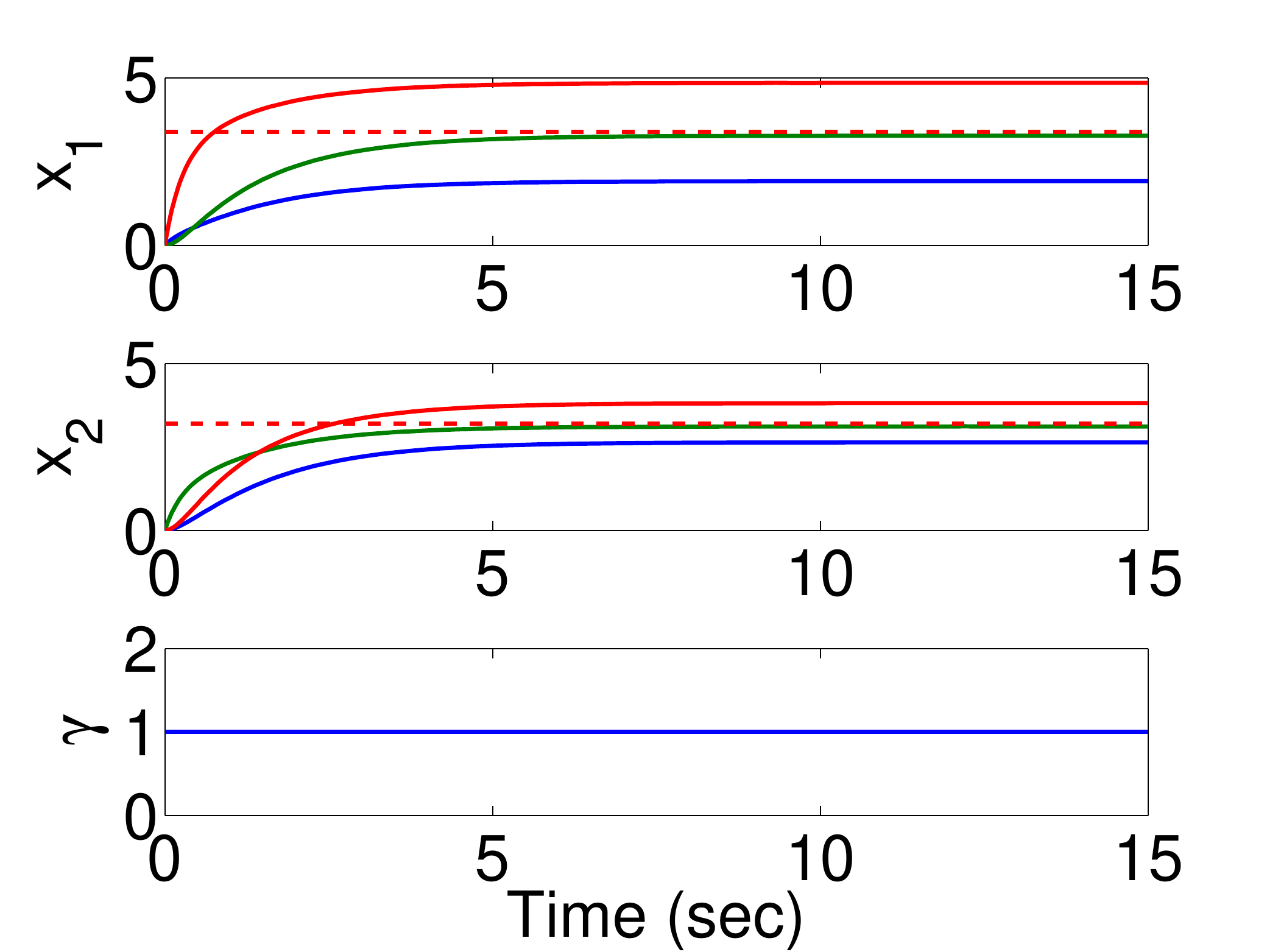} &
\includegraphics[width=.45\linewidth]{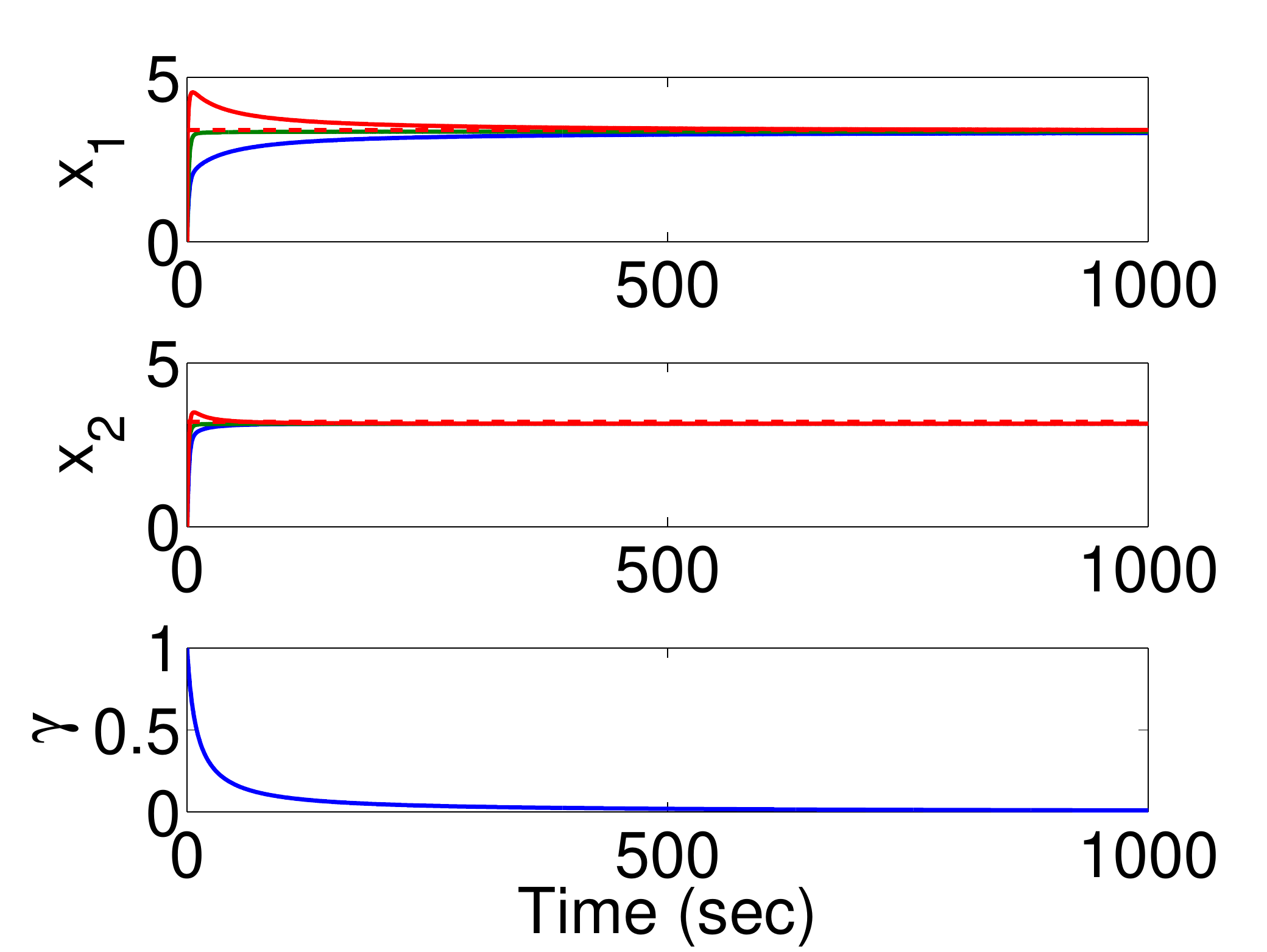} 
\end{array}$
\end{center}
\caption{This figure shows the result of optimizing using consensus for the problem given in (\ref{eq:convex_example}) for both a constant and fading value for $k_G$ on the left and right respectively}%
\label{fig:convex_P}
\end{figure*}

\begin{remark}
In presenting examples throughout the remainder of the paper, the results from both a constant and a diminishing gain will be shown.  We do this instead of trying to tune the ``stopping'' criteria given in \cite{Nedic2009}.  The result of a constant gain will emphasize the possible convergence rate and a diminishing gain will emphasize the ability to reach optimality.
\end{remark}
\section{PI distributed optimization}  
\label{sec:PIOptimization}

In Sections \ref{sec:dualDecomposition} and \ref{sec:consensus}, dual decomposition and the consensus method for distributed optimization were introduced and the parallel to integral and proportional control laws was seen.  In this section, we show that these two methods can be combined to create a new distributed optimization method which is guaranteed to converge to the global minimum, much like integral control can be added to proportional control to achieve zero steady-state error with good convergence properties.

This section begins by developing the PI distributed optimization method and proving that it converges to the global minimum.  The relationship to PI control is then discussed and the  example of the previous two sections is finished.

\subsection{PI distributed optimization algorithm}
\label{ssec:pi_optimization}
The PI distributed optimization algorithm is formed by noting that the dual-decomposition method discussed in Section \ref{sec:dualDecomposition} shares similar structure with the consensus method discussed in Section \ref{sec:consensus}.  Each has a gradient term along with an additional term added to enforce equality between agents.  Dual-decomposition guarantees convergence to the goal, but has an undesirable transient, oscillatory behavior.  On the other hand, the consensus method does not converge under constant gains, but has a much more damped transient response.  Therefore, we join the two methods in a desire to achieve the benefits of each.

Combining equations (\ref{eq-dual_decomp_z_dynamics}) and (\ref{eq-dual_decomp_mu_dynamics}) with (\ref{eq:consensus_based}), the aggregate dynamics can be expressed as 
\begin{equation} 
\begin{split}
  \dot{z}  =  -k_G\frac{\partial f}{\partial z}^T   - k_P \Lbig z  - k_I'\D \mu  \\
  \dot{\mu}  = k_I'\D^T z.
\end{split}
\label{eq:PI_dynamics}
\end{equation}

Similarly, (\ref{eq:x_singleAgent_I}) and (\ref{eq:mu_singleAgent_I}) can be combined with (\ref{eq:consensus}) to get a distributed implementation as follows:
\begin{equation}
\dot{x}_i = -k_G \frac{\partial f_i}{\partial x}^T(x_i) -  k_P \sum_{j \in \N_i}  (x_i - x_j) - k_I' \sum_{j \in \N_i} \mu_i^j.
\end{equation}
\begin{equation}
\dot{\mu}_i^j = k_I'(x_{i} - x_{j})
\end{equation}
where we again define $\mu_i^j$ as in (\ref{eq:mu_singleAgent_I}).  As in Sections \ref{sec:dualDecomposition} and \ref{sec:consensus}, {\it the only information exchange required between agents is the exchange of the state vectors between neighboring agents}.

To show convergence to the global minimum, we give the following two theorems.

\begin{theorem} \label{th-PI_global_stable}
Given Assumptions \ref{as:strictlyconvex}, \ref{as:solutionexistence}, and \ref{as:ConnectedGraph} as well as the dynamics in (\ref{eq:PI_dynamics}), the saddle point $(\dot{z}, \dot{\mu}) = (0,0)$ is globally asymptotically stable.
\end{theorem}
\begin{proof}
The same proof can be used as was used in Theorem \ref{th-dual_decomp} with two modifications.
\begin{enumerate}
\item [(1)] $H(z) = k_G \frac{\partial ^2 f}{\partial z^2} + k_P \Lbig$, but $H(z) \succ 0$ still holds.
\item [(2)] $\ddot{z} = -k_G\frac{\partial ^2 f}{\partial z^2} \dot{z} - k_P\Lbig \dot{z} - k_I'\D \dot{\mu}$ which when $\dot{z} = 0$ still simplifies to $\ddot{z} = -k_I'\D \dot{\mu}$
\end{enumerate}
\end{proof}

\begin{theorem} \label{th-PI_global-min}
Given Assumptions \ref{as:strictlyconvex}, \ref{as:solutionexistence}, and \ref{as:ConnectedGraph} as well as the dynamics in (\ref{eq:PI_dynamics}), the saddle point $(\dot{z}, \dot{\mu}) = (0,0)$ corresponds to the global minimum. 
\end{theorem}
\begin{proof}
The same proof can be used as was used in Theorem \ref{th-dual_global_min} by noting for a feasible solution, $\Lbig z = 0$.  This will give the same equation for $\dot{z}$ as given in (\ref{eq:minCond_from_zdot}).
\end{proof}

The proofs of Theorems \ref{th-PI_global_stable} and \ref{th-PI_global-min} basically show that adding the consensus term does not break the convergence properties of the dual-decomposition method of Section \ref{sec:dualDecomposition}, but do nothing to speak of the benefit of adding the consensus term.  To see the benefit of the consensus term, consider the following problem:
\begin{equation}
\min_z k_G f(z) + \frac{k_P}{2} z^T \Lbig z.
\label{eq:PIProblem}
\end{equation}
$$
\mbox{s.t. } k'_I \D^T z = 0
$$
This is the same problem as given in (\ref{eq:constrained_opt}), but with the addition of a term proportional to the square of the constraint (recall $\D \D^T = \Lbig$).  Adding the square of the constraint is known as the augmented Lagrangian method, which has been shown to add dampening to the dual optimization problem, improving convergence, (see \cite{Boyd2011} for a discussion and analysis of the augmented Lagrangian).

Following the same method to develop dynamic update laws as in Section \ref{sec:dualDecomposition}, the following dual optimization problem would be solved:
\begin{equation} 
 \max_{\mu} \min_{z} \Bigl( k_G f(z) + k'_I z^T \D \mu + \frac{k_P}{2} z^T \Lbig z  \Bigr),
\end{equation}
with the resulting dynamics being the same as (\ref{eq:PI_dynamics}).  Thus, adding in a consensus term corresponds to modifying the problem to solve the augmented Lagrangian, producing the desired dampening effect without modifying the guarantee of convergence.

\subsection{Connections to PI control}
As with the previous two distributed optimization techniques, we note the similarity of this distributed optimization framework with a PI control framework. The Lagrange multiplier, $\mu$, can be expressed in the same form as done in (\ref{eq:mu_integral_control}).  Thus, the following expression for $\dot{z}$ can be obtained:
\begin{equation}
\dot{z}(t) = -k_G \frac{\partial f}{\partial z}^T - k_I \Lbig \int_{t_0}^t z(\tau)d\tau - k_P \Lbig z(t).
\label{eq:PI_control}
\end{equation}
This is the same equation that was derived for a PI control law in Section \ref{ssec:constrained_opt}.  We can therefore expect to see properties of PI control such as increased overshoot resulting from decreased dampening of the proportional control, zero steady-state error due to the integral term (which has already been proved), and faster settling time than pure integral control, e.g. \cite{GF2001}.

While there exist many distributed optimization techniques, e.g. \cite{Palomar2010, Wei2012, Kvaternik2012, Wang2010, Wang2011a, Terelius2011}, it is important to note the similarity of the method in this section to that presented in \cite{Wang2010} and extended in \cite{Wang2011a, Gharesifard2012}.  While the development of the algorithm in \cite{Wang2010} is different than the development in this paper, it can be expressed as using the augmented Lagrangian to solve the following problem:
\begin{equation}
	\min_z f(z),
\end{equation}
$$
\mbox{s.t. } \Lbig z = 0
$$
where the resulting dynamics can be expressed as
\begin{equation}
\dot{z} = - \frac{\partial f}{\partial z}^T(z(t)) - \Lbig z(t) - \Lbig \mu(t),
\end{equation}
\begin{equation}
\dot{\mu} = \Lbig z,
\end{equation}
and now $\mu \in \R^{Nn}$ as opposed to $\mu \in \R^{Mn}$ as before.  The only difference between this method and the method we have developed is simply that the constraint is expressed in terms of the graph Laplacian instead of the incidence matrix.  This would result in an equation similar to (\ref{eq:PI_control}), except with an $\Lbig^2$ term instead of an $\Lbig$ term in front of the integral.  

While this may seem like a small difference, due to the fact that we have utilized dual-decomposition in the development of the integral term, {\it we form a PI distributed optimization technique which requires half of the communication} that the technique developed in \cite{Wang2010} requires.  This can be seen from the fact that the incidence matrix, used in dual-decomposition, allows each agent to update the necessary values of $\mu$ using only local information.  However, using the Laplacian matrix to express the constraint forms an $\Lbig^2$ term which requires that either each agent knows their neighbors' neighbors states or each neighbor must additionally communicate $\mu_i$ at each optimization step. 

\subsection{Example}

\begin{figure*}[!t]
\begin{center}
\includegraphics[width=.4\linewidth]{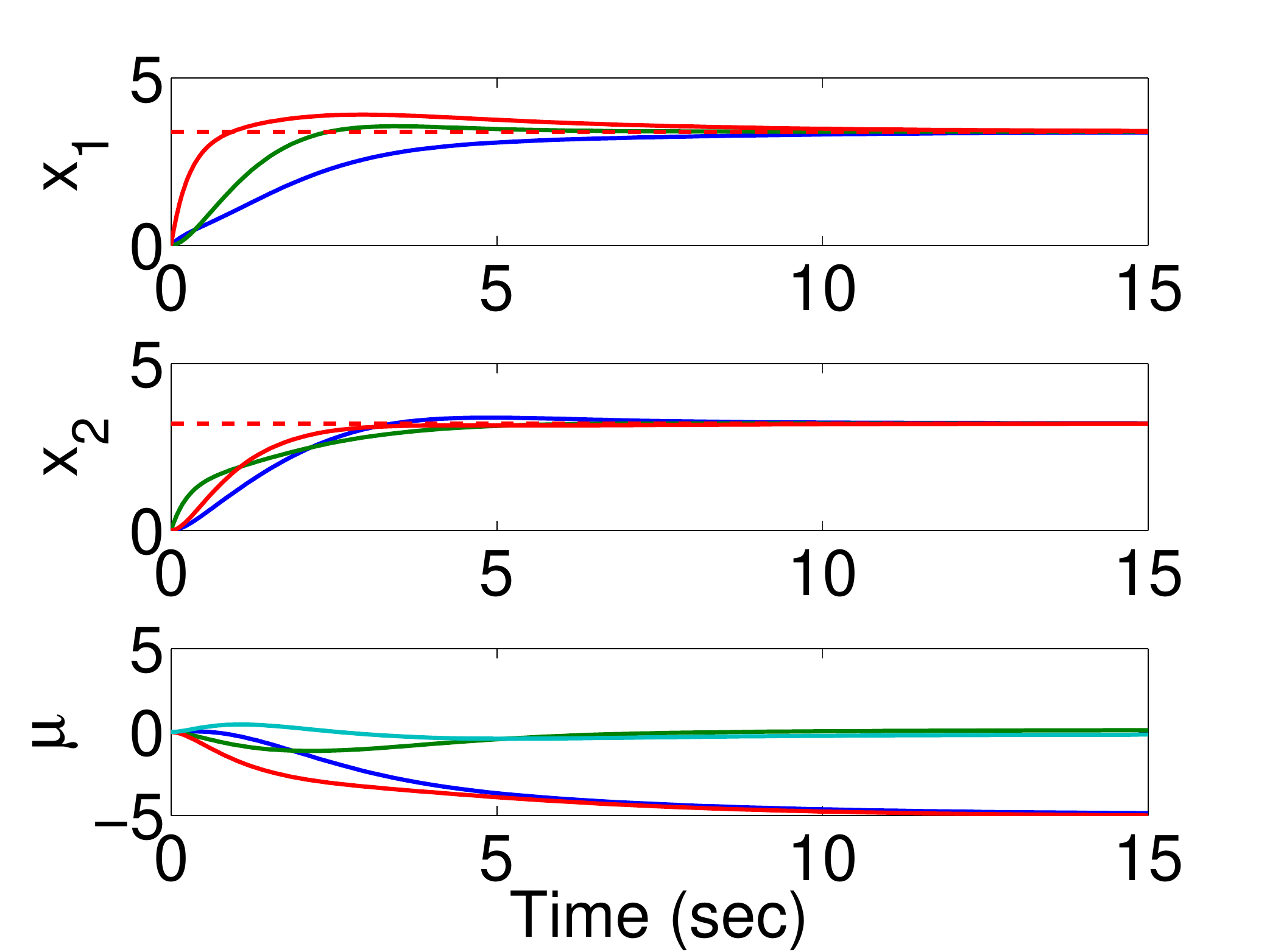} 
\caption{This figure shows the results from the convex optimization example using PI distributed optimization}%
\label{fig:convex_PI}
\end{center}
\end{figure*}

We continue the example in (\ref{eq:convex_example}) using the newly derived dynamics.  In Figure \ref{fig:convex_PI}, it is apparent that the PI optimization is able to achieve zero error while converging quickly and with little oscillation.  Furthermore, Table \ref{tbl:convex} shows that settling time and overshoot are in between the values of pure proportional and pure integral control, as expected.  These attributes will be emphasized in the examples in the following sections as more complex problems are presented.

\section{Scalable multi-agent formulation}  
\label{sec:scalability}
Up until this point, we have presented the algorithms in terms of a framework where each agent keeps its own version of the entire state vector as done in previous works, e.g. \cite{Kvaternik2012, Wang2010, Nedic2009, Terelius2011}.  This is not necessary if some of the agents' individual costs do not depend upon all of the elements of the parameter vector being optimized.  An example of this will be shown at the end of the section where each agent introduces more parameters to be optimized, typical in multi-robot scenarios, e.g. \cite{Keviczky2006, Droge2013, Rantzer2009, Dunbar2006}.  However, each agents' cost depends solely on the parameters introduced by its neighbors.  In such a situation, it is not necessary for each agent to keep track of the entire parameter vector and, in fact, doing so is not scalable to large numbers of agents.  

In this section, we address this in a similar fashion to \cite{Wang2011a} and show that it fits quite naturally into the framework of the previous sections.  First, it is shown that even with the reduction of parameters the previous theorems still hold.  Then, the reduction of parameters will lead to a slight reformulation of the PI distributed optimization algorithm.  Finally, we end this section with an example where drastic improvement in convergence is achieved by reducing the number of variables that each agent must maintain. 

\subsection{Eliminating unneeded variables}
When each agent does not have an opinion about a parameter in the parameter vector, the problem can be simplified to eliminate redundancies.  Similar to \cite{Wang2011a}, let $I_j = \{ i| f_i$ depends on the element $j\}$ be the set of agents which depend on element $j$ with cardinality $N_j = |I_j|$.  As agents no longer needs to keep track of the entire vector, the definition of $z_j$ needs to be slightly modified to $z_j \triangleq vec[x_{ij}]_{i \in I_j} \in \mathbb{R}^{N_i}$, a subset of the elements originally contained in $z_j$.  Now, the aggregate vector can be defined as $z = \begin{bmatrix}
z_1^T & ... &z_n^T
\end{bmatrix}^T \in \mathbb{R}^{N_1 + ... + N_n}$.

Let the induced subgraphs, $\G_i(\V_i, \E_i)$, be defined as $\V_i = \{ v_j \in \V | j \in I_i \} \subseteq \V$ and $\E_i = \{(v_i, v_j) \in \E | v_i, v_j \in \V_i \}$.  Finally, the following assumption is made to allow for convergence
\begin{assumption}
$\G_i$ is connected $\forall i \in \{1, ..., N \}$. 
\label{as:subgraph_connected}
\end{assumption}
Note that, given Assumption \ref{as:ConnectedGraph}, Assumption \ref{as:subgraph_connected} is not limiting.  If there exists $i$ s.t. $\G_i$ is not connected, one needs only to extend $\G_i$ to contain nodes originally in $\G$ that will connect the different connected components of $\G_i$.  

Along this same line of reasoning, we briefly touch upon a topic of study which is out of the scope of this paper, but worth mentioning.  There may be simple cases in which choosing $\G_i$ such that it is connected with the smallest number of vertices possible will not result in the fastest convergence to the global minimum.  There has been much work done on the convergence of the consensus equation and the network topology plays a key role in determining the convergence rate \cite{Mesbahi2010, Olfati-Saber2004, Jadbabaie2003}.   Therefore, to achieve the fastest performance, selection of the sub-graph for each variable could be more complicated than simply choosing the minimally connected sub-graph.

In any case, given $\G_i$, the corresponding incidence matrix, $D_i \in \mathbb{R}^{N_i \times M_i}$, where $M_i = |\E_i|$, and graph Laplacian, $L_i \in \mathbb{R}^{N_i \times N_i}$, can be defined.  This allows for the definition of the aggregate matrices $\D \triangleq diag(D_1, ..., D_n)$ and $\Lbig = diag(L_1, ..., L_n)$.  These aggregate matrices will continue to exhibit the same properties mentioned in Section \ref{ssec:networkedSystems} as they can still be expressed as $n$ connected components of a graph.  The only difference is that the connected components do not have the same structure.  As these properties still hold, Theorems \ref{th-dual_decomp} through \ref{th-PI_global-min} will also hold using the newly defined augmented matrices and addition of Assumption \ref{as:subgraph_connected}.

\subsection{Distributed implementation}
While the aggregate dynamics of the multi-agent system can be expressed without any change, the dynamics executed by each agent will change slightly due to the fact that each variable in the parameter vector will have a different set of agents that are maintaining a version of it.  We express the dynamics of a single variable as follows:
\begin{equation}
\dot{x}_{ij} = -k_G \frac{\partial f_i}{\partial x_{ij}} - k_P \sum_{k \in \{\N_i \cap I_j \} } (x_{ij} - x_{kj}) - 
k_I'\sum_{k \in \{\N_i \cap I_j \} } \mu_{i,j}^k
\end{equation}
\begin{equation}
\dot{\mu}_{i,j}^k = k_I' (x_{ij} - x_{kj}).
\end{equation}
Note that the algorithms in Sections \ref{sec:dualDecomposition} and \ref{sec:consensus} can be achieved by setting $k_P = 0$ and $k_I' = 0$ respectively.  Again, we see that {\it each agent is able to execute this algorithm using local information and only communicating its version of the parameters being optimized with its neighbors}.

\subsection{Ring example}
We now present an example in which scaling down the number of parameters that each agent worries about drastically improves the performance of the system.  Consider the ``Ring'' network depicted in Figure \ref{fig:Ring} where each agent can communicate with agents to each side.  In this example, each agent has a variable that ``belongs'' to it and it wants to balance having its value be close to its neighbors' value as well as a nominal value.  This can be expressed in the form of the following quadratic cost:
\begin{equation}
	f_i = (x_{i,i-1} - x_{ii})^2 + (x_{ii} - x_{d_i})^2 + (x_{ii} - x_{i,i+1})^2
	\label{eq:simlified_cost}
\end{equation}
where $x_{d_i} = i$ is the desired value.

\begin{figure*}[!t]
\begin{center}
\includegraphics[width=0.3\linewidth]{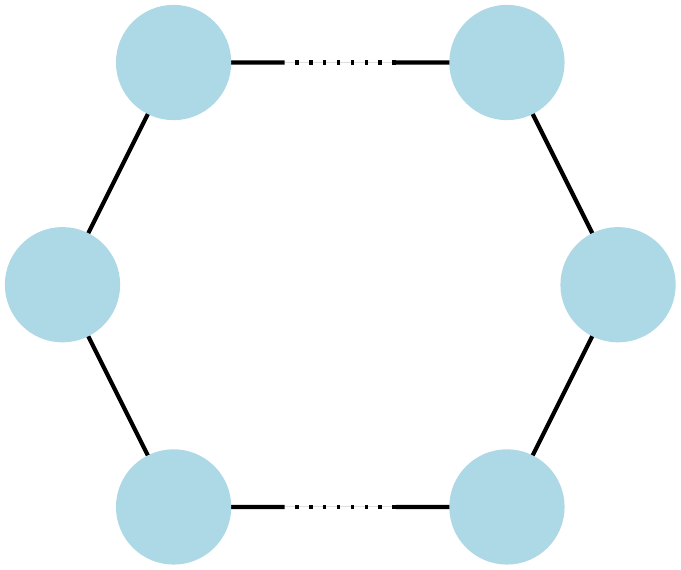} 
\end{center}
\caption{This figure depicts the ``Ring'' network structure used in Section \ref{sec:scalability} }%
\label{fig:Ring}
\end{figure*}

Note that for the formulation in Sections \ref{sec:dualDecomposition}, \ref{sec:consensus}, and \ref{sec:PIOptimization}, each agent would have had to keep track of $N = 20$ variables, corresponding to the aggregate state vector having $400$ elements.  However, this is greatly reduced by following the formulation in this section.  Each agent will only need to keep track of $3$ variables with a total of $60$ variables in the aggregate state vector.  

The results of both representations of the state can be seen in Figure \ref{fig:Simplified_full_simple} and Tables \ref{tbl:simplified_full} and \ref{tbl:simplified_simple}.  Significant improvement can be seen across the board in terms of settling time for reducing the number of variables.  Moreover, the overshoot is drastically improved for both the I and PI distributed optimization methods.  Related to overshoot, it is seen in Figure \ref{fig:Simplified_full_simple} that the oscillation is drastically reduced for dual-decomposition.  

One final observation about the performance of the PI distributed optimization technique is noteworthy.  This example demonstrates the performance of the system when a larger number of variables is in question.  We see in Table \ref{tbl:simplified_full} that the PI distributed optimization significantly outperforms the other methods in terms of convergence.  Moreover, there is a drastic improvement over the dual-decomposition method in terms of overshoot and oscillation as well as an improvement over the consensus method in terms of steady-state error.

Again, we emphasize that this is an extreme example meant to demonstrate the possible utility of reducing the number of variables that each agent deals with.  Conclusions should not be drawn beyond the notion that this may be beneficial as there may be instances in which scaling back as much as possible would not be beneficial.

%

\begin{figure*}
\begin{center}$
\begin{array}{ccc}
\includegraphics[width=.31\linewidth]{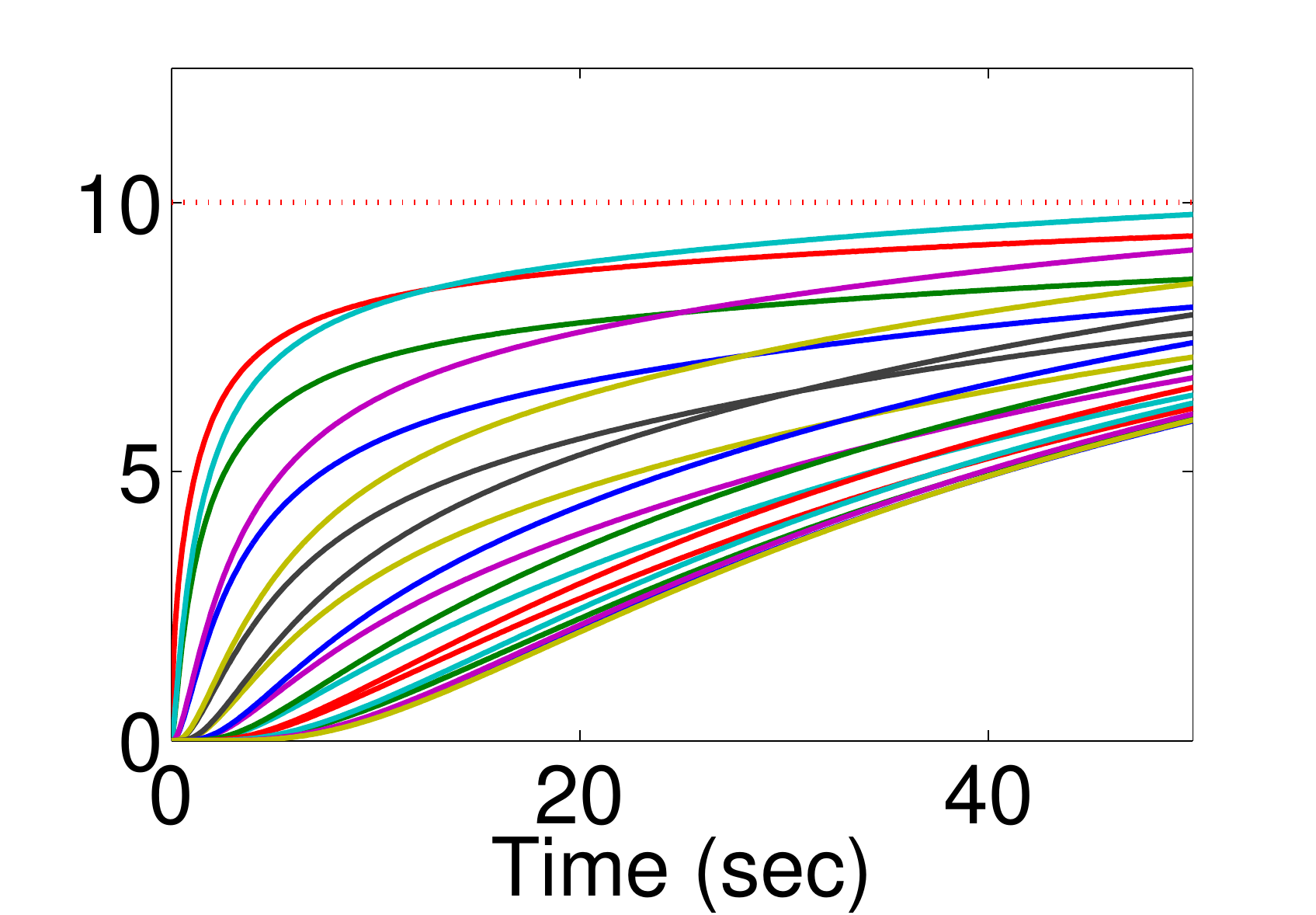} &
\includegraphics[width=.31\linewidth]{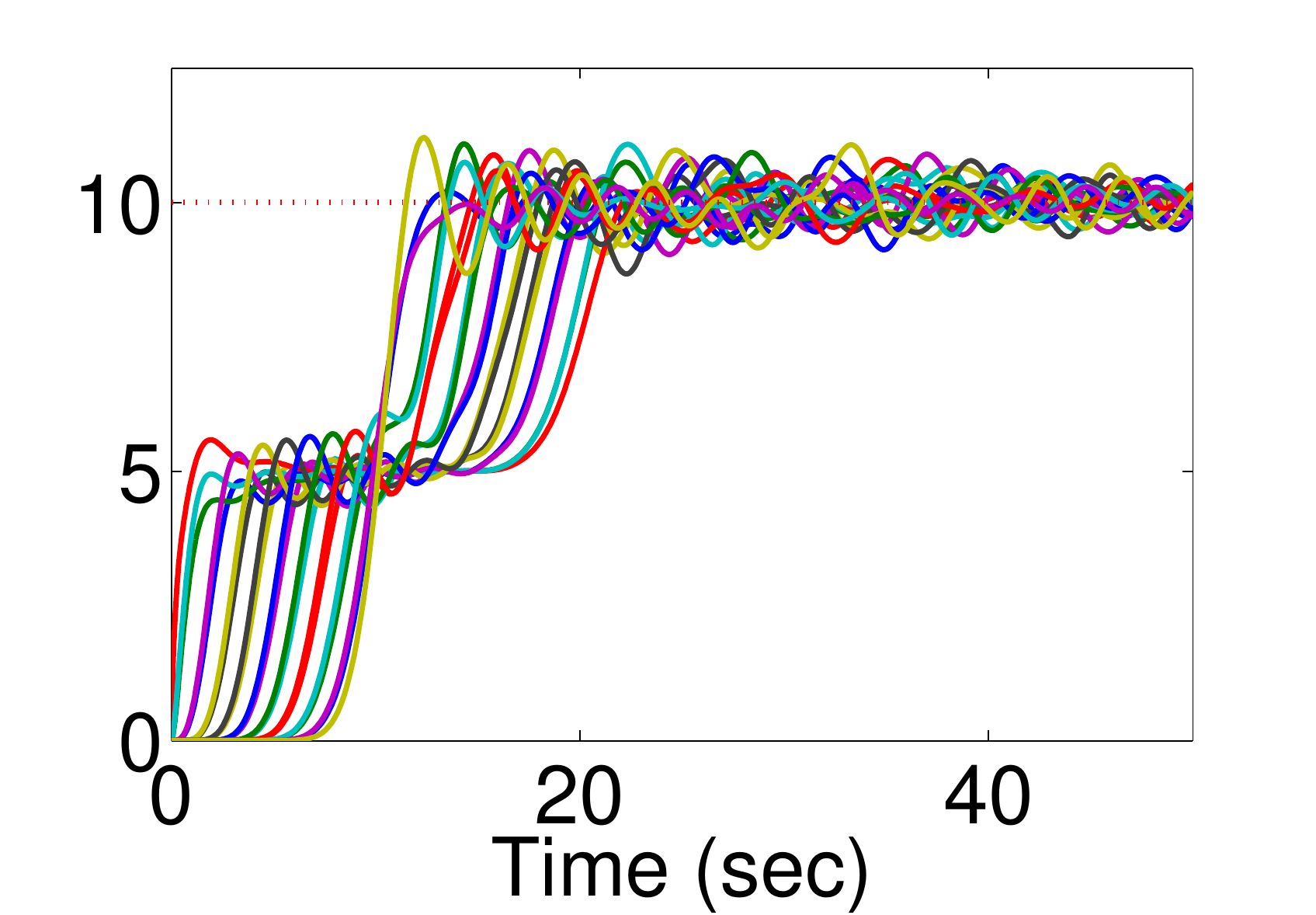} &
\includegraphics[width=.31\linewidth]{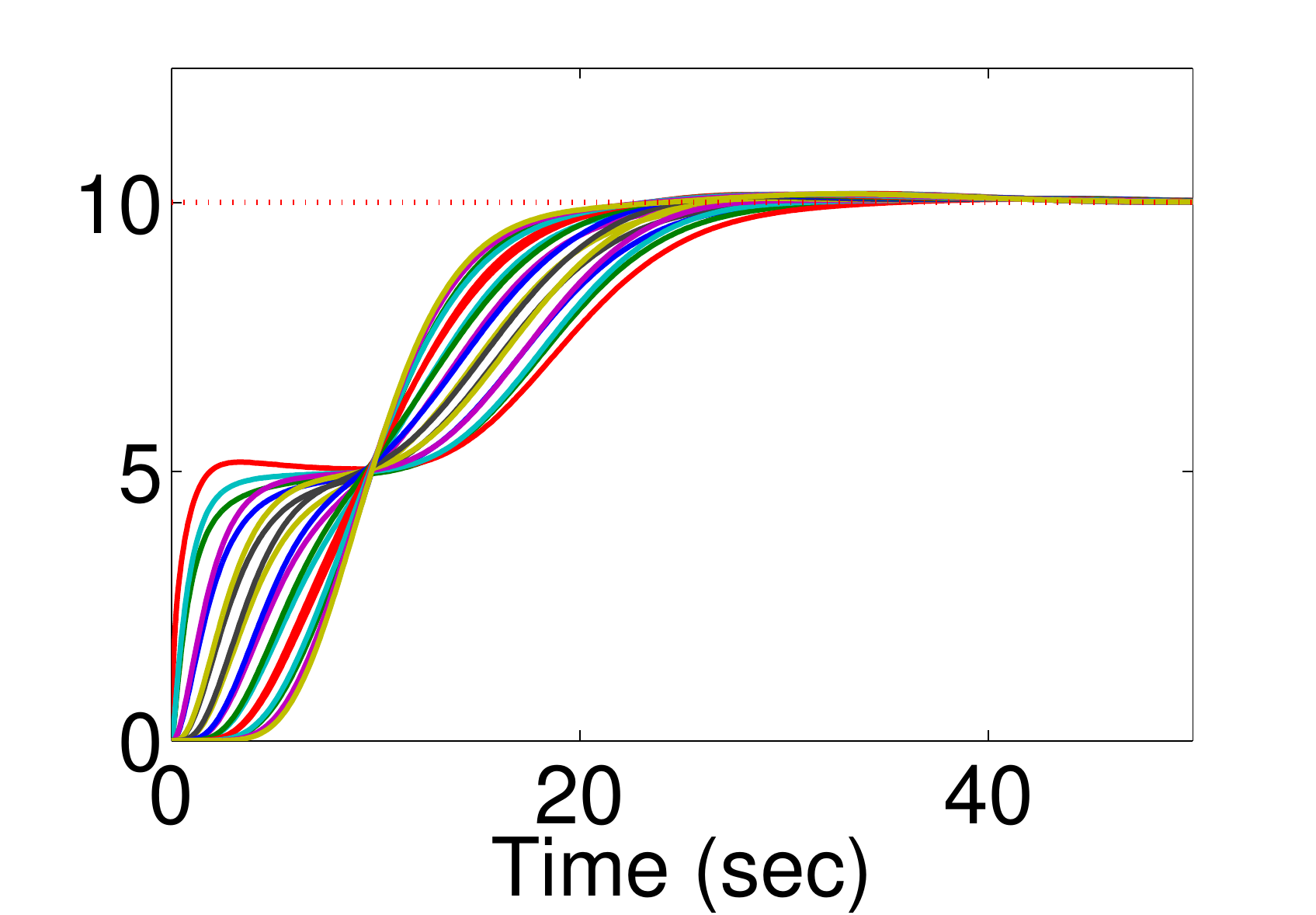} \\
\includegraphics[width=.31\linewidth]{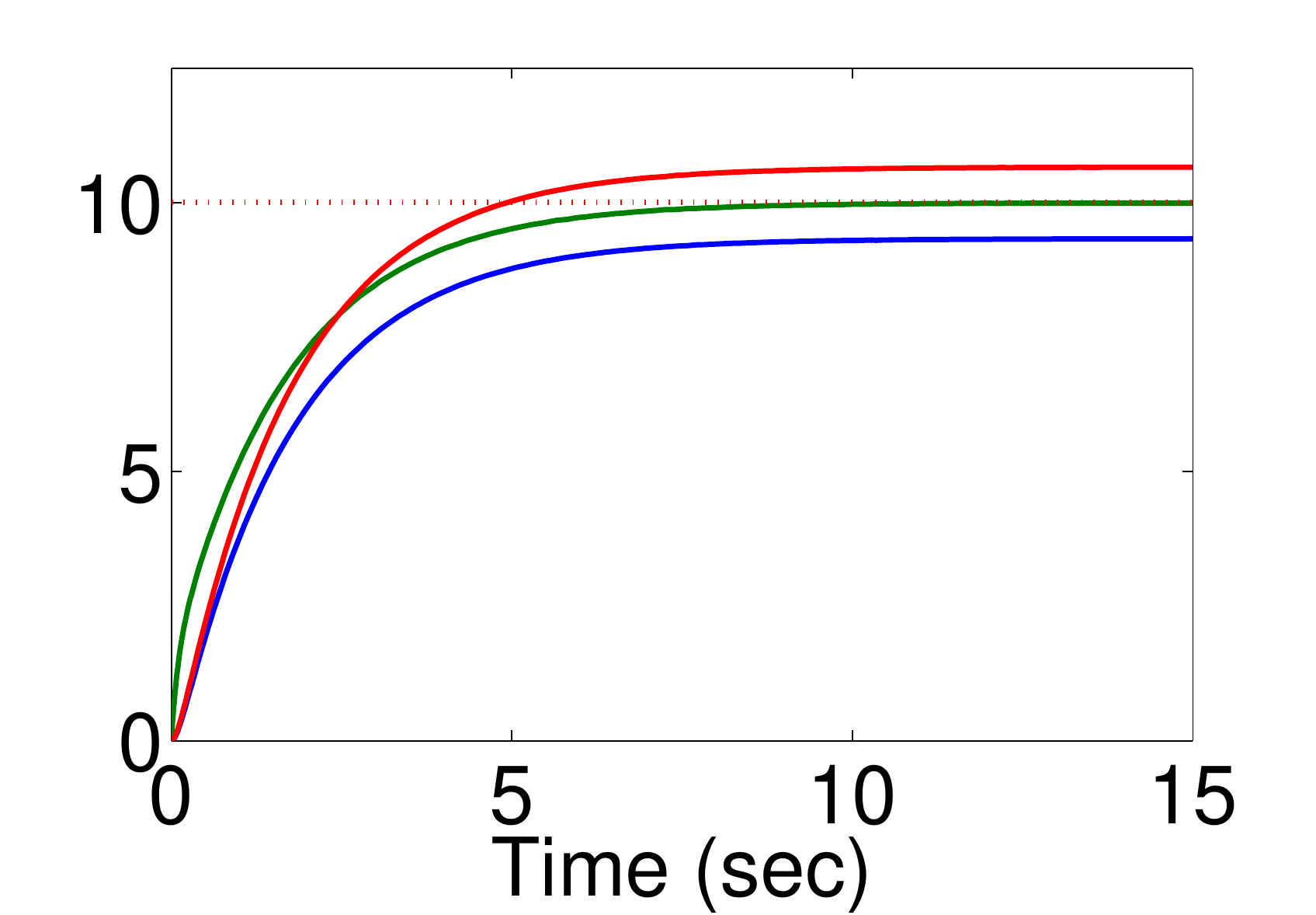} &
\includegraphics[width=.31\linewidth]{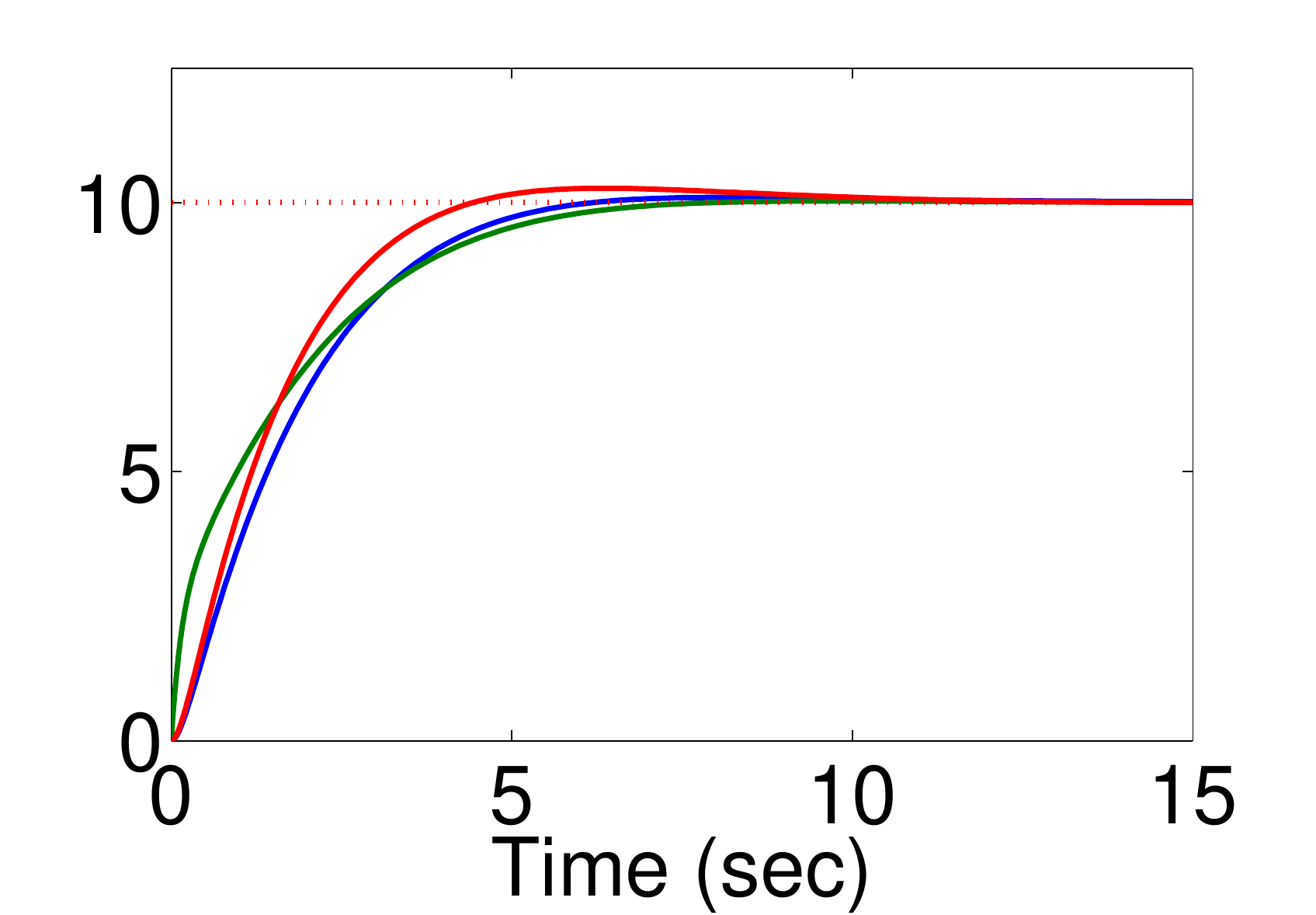} &
\includegraphics[width=.31\linewidth]{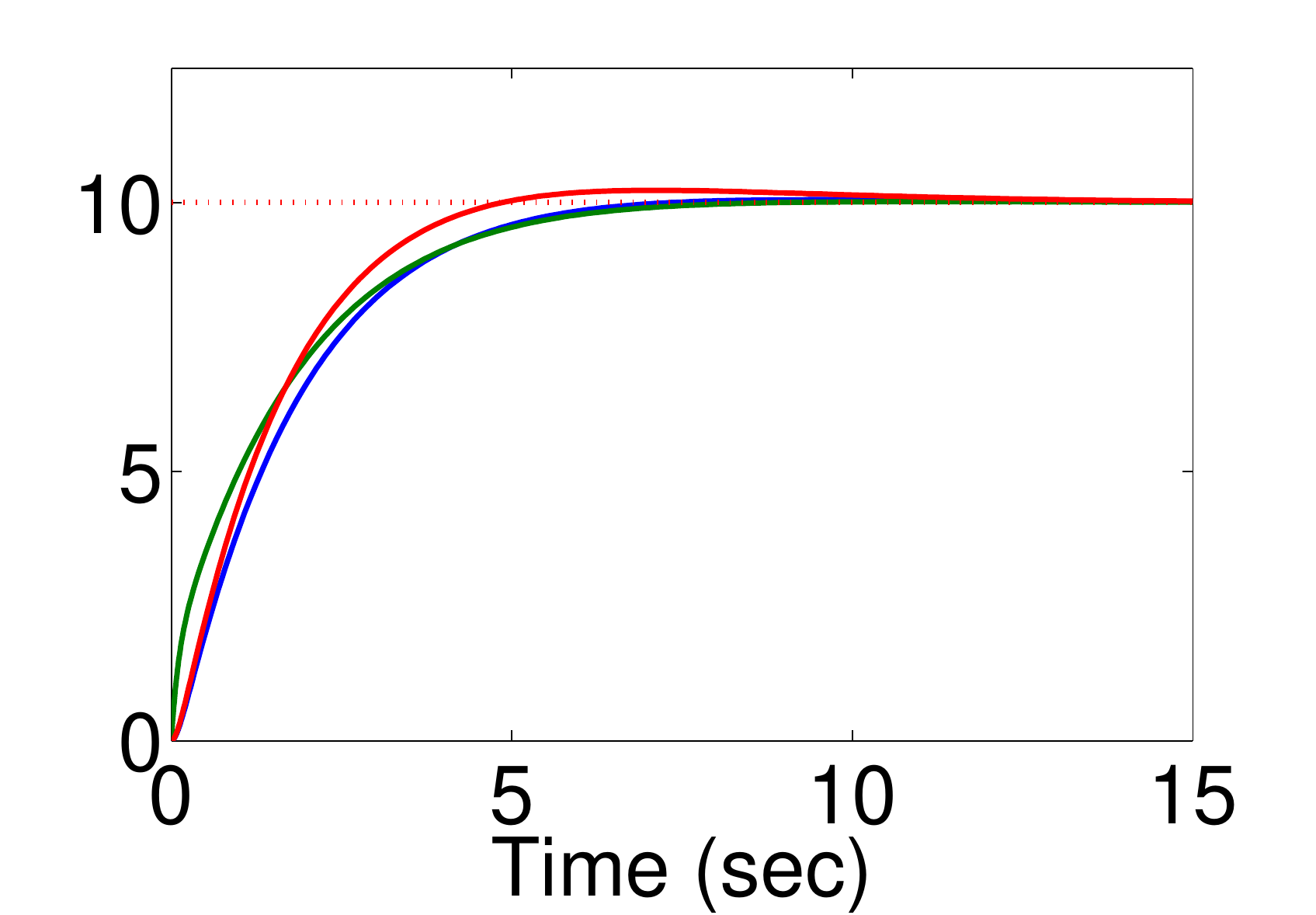} 
\end{array}$
\end{center}
\caption{This figures shows the results of applying the formulation of Sections \ref{sec:dualDecomposition}, \ref{sec:consensus}, and \ref{sec:PIOptimization} on the top row and \ref{sec:scalability} bottom row to solve the problem in (\ref{eq:simlified_cost}).  The left, middle, and right images of each row correspond to consensus, dual-decomposition, and PI distributed optimization techniques.  The results shown are for variable 10.  The solutions in the top row require 20 versions of this variable to converge to the optimal value where the solutions in the bottom row require only 3.}%
\label{fig:Simplified_full_simple}
\end{figure*}

\begin{table}[!t]
\centering
\caption{The results of performing proportional, integral, and PI distributed optimization with each agent optimizing over the full state vector } 
\label{tbl:simplified_full}
\begin{tabular}{|c|c|c|c|c|}
\hline 
& P: $\gamma = 1$ & P: $\gamma = \frac{1}{1 + .1t}$ & I & PI \\ 
\hline 
$M$ & 0.1\% & 0.12\% & 37.5\% & 7.9\% \\ 
\hline 
$t_{10}$ & 120.8 & 659.42 & 115.28 & 29.78 \\ 
\hline 
$t_1$ & 226.58 & 4884.8 & 542.71 & 83.02 \\ 
\hline 
\% error & 55.4\% & 0.92\% & 0\% & 0\% \\ 
\hline 
\end{tabular}
\end{table}

\begin{table}
\centering
\caption{The results of performing proportional, integral, and PI distributed optimization with each agent optimizing over a subset of the state vector }
\label{tbl:simplified_simple}
\begin{tabular}{|c|c|c|c|c|}
\hline 
 & P: $\gamma = 1$ & P: $\gamma = \frac{1}{1 + .1t}$ & I & PI \\ 
\hline 
$M$ & 0.1\% & 35.15\% & 7.12\% & 4.51\% \\ 
\hline 
$t_{10}$ & 5.2 & 82.85 & 6.12 & 6.03 \\ 
\hline 
$t_1$ & 9.47 & 692.57 & 12.78 & 12.33 \\ 
\hline 
\% error & 57.48\% & 5.3\% & 0\% & 0\% \\ 
\hline 
\end{tabular} 
\end{table}

\section{Conclusion}  
\label{sec:conclusion}
We have developed a new, PI distributed optimization method through the combination of dual decomposition and the consensus method for distributed optimization.  This has been done by noting the similarity of the methods when considering the underlying constrained optimization problem.  This new method is able to achieve desirable properties from both of the previous methods.  Namely, faster convergence and dampening due to the proportional term, originating from the consensus based method, and zero steady-state error from the integral term, originating from dual-decomposition.  The method was also modified to allow agents to maintain only the variables they care about, with an example showing drastic improvement in convergence times.   

\section*{Funding}
\small{The work by M. Egerstedt was funded by The Air Force Office of Scientific Research through grant number [2012-00305-01]}

\setstretch{1.0}

\bibliographystyle{IEEEtran}
\bibliography{Submission}

\end{document}